\documentclass{elsarticle}
\usepackage[a4paper, total={6in, 10in}]{geometry}
\usepackage{booktabs}
\usepackage{float}
\usepackage{amsmath}
\usepackage{amsthm}
\usepackage{amsmath,amssymb,amsfonts,mathrsfs,hyperref,microtype,amsthm}
\usepackage{graphicx}
\usepackage{graphicx,mathabx}
\usepackage{enumitem}  

\graphicspath{ {images/} }
\numberwithin{equation}{section}

\newtheorem{proposition}{Proposition}[section]

\newtheorem{example}{Example}
\newtheorem*{example*}{Example}
\newcommand{\BS}{\boldsymbol}

\newcommand{\Rmnum}[1]{\expandafter\@slowromancap\romannumeral #1@}

\journal{}
\makeatletter
\def\ps@pprintTitle{%
   \let\@oddhead\@empty
   \let\@evenhead\@empty
   \def\@oddfoot{\reset@font\hfil\thepage\hfil}
   \let\@evenfoot\@oddfoot
}

\begin{document}

\begin{frontmatter}
\author[a]{Helmi Shat\corref{cor1}}
\cortext[cor1]{Corresponding author}
\ead{hshat@ovgu.de}
\author[a]{Rainer Schwabe}
\ead{rainer.schwabe@ovgu.de}
\address[a]{\small Institute for Mathematical Stochastics, Otto-von-Guericke University Magdeburg, \\ \small Universit\"atsplatz 2,
39106 Magdeburg, Germany }

 \title{Experimental Designs for Accelerated
Degradation Tests Based on Linear Mixed Effects Models}

\begin{abstract}
Accelerated degradation testing has considerable significance in reliability engineering due to its ability to provide accurate estimation of lifetime charachteristic of highly reliable systems within a relatively short testing time period.
The measured date from particular experiments at high stress conditions are extrapolated, through a technically reasonable statistical model, to obtain estimates of certain reliability properties under normal use levels.
In this work we consider  repeated measures accelerated degradation tests with multiple stress variables, where the degradation paths are assumed to follow a linear mixed effects model which is quite common in settings when repeated measures are made.
We derive optimal experimental designs for minimizing the asymptotic variance for estimating the median failure time under normal use conditions when the time points for measurements are fixed in advance.  
\end{abstract}


\end{frontmatter}

\section{Introduction}
\label{sec-introduction}

Industrial needs for sustainable and highly reliable systems have motivated corresponding manufacturers to design and manufacture products that can operate without failure for years or even decades. 
As a consequence, manufacturers are demanded to provide their customers with accurate information about the reliability of their products. 
However, when the products get more reliable, it becomes more difficult or even unfeasible to assess a sufficient amount of lifetime data on the basis of traditional reliability testing in order to accurately estimate characteristics of the lifetime distribution of the products because failure or fatigue can hardly be observed under normal use conditions in a reasonable time period for testing.
As an alternative, for highly reliable and enduring products, Accelerated Degradation Tests can be utilized to provide sufficient information on the deterioration of the products to obtain a sufficiently accurate estimate of lifetime properties within a relatively short testing time period. 
In Accelerated Degradation Testing products are tested at various elevated stress levels (for e.\,g.\ temperature, voltage, or vibration).
The resulting data are then extrapolated, through a technically reasonable statistical model, to obtain estimates of lifetime characteristics under normal use conditions. 
The precision of the estimates is influenced by several factors, such as the number of units tested, the duration of the testing period, the frequency of measurements, and, on particular, the choice of the stress levels to which the units are exposed. 

A vast amount of literature is devoted to the analysis of Accelerated Degradation Tests, see, for example, \cite{meeker2014statistical} for a comprehensive survey on various approaches in the literature used to assess reliable information from degradation data.
In more detail, \cite{bagdonavicius2001accelerated} present models and methods of statistical analysis for Accelerated Degradation Tests and further references can be found there. 
As additional sources, \cite{1505040} provides an extensive list of references related to accelerated test planning and \cite{limon2017literature} review prominent methods for statistical inference and optimal design of accelerated testing plans. 
There different types of test planning strategies are categorized according to their merits and drawbacks and research trends are provided. 
\cite{li2006design} presents an analytical method for the optimum planning of Accelerated Degradation Tests with an application to the reliability of Light-Emitting Diodes. 
There the author states that the variability of the measured units have a substantial impact on the accuracy of estimation. 
Therefore these random effects should be encountered in the choice of the experimental settings for the Accelerated Degradation Tests. 
Based on the observation that ignoring the variability in the normal use conditions may lead to significant prediction errors, \cite{liao2006reliability} extend Accelerated Degradation Test models to predict field reliability by considering variations in the stress levels by considering a degradation process  represented by a Brownian motion with linear drift via a stochastic differential equation. 
\cite{wang2017optimal} propose a $M$-optimality criterion for designing constant stress Accelerated Degradation Tests when the degradation path can be represented by an inverse Gaussian process with covariates and random effects. 
This criterion focuses on a degradation mechanism equivalence rather than on the evaluation precision or the prediction accuracy which are usually employed in traditional optimization criteria. 
Those authors prove that, with a slightly relaxed requirement of prediction accuracy, the obtained optimum designs minimize the dispersion of the estimated acceleration factor between the normal stress level and a higher accelerated stress level. 
Wiener processes (Brownian motions) are intensively used to represent degradation paths in Accelerated Degradation Testing, see \cite{ye2015new} and \cite{guan2016objective}. 
For instance, \cite{lim2011optimal} develop optimal Accelerated Degradation Test plans assuming that the constant stress loading method is employed and the degradation characteristics follows a Wiener process. 
These authors determine the test stress levels and the proportion of test units allocated to each stress level such that the asymptotic variance of the maximum likelihood estimator of a particular quantile of the lifetime distribution at the normal use condition is minimized. 
In addition, compromise plans are also developed for checking the validity of the relationship between the model parameters and the stress variable.
In a case study for random effects in degradation of semiconductors, \cite{lu1997statistical} propose a repeated measurements model with random regression coefficients and a standard deviation function for analyzing linear degradation data. 
The authors utilize several large sample interval estimation procedures to estimate the failure time distribution and its quantiles. 

On the other hand, the general theory of optimal design of experiments is well developed in the mathematical context of approximate designs which allow for analytical solutions (see e.\,g.\ \cite {silvey1980optimal} or \cite{d1f763d8224548d6a7da9f5719fdc086}). 
In addition, \cite{schwabe1996optimum} deals with the theory of optimal designs for multi-factor models which can be used here to treat more than one stress variable and the choice of time plans simultaneously under various interaction structures. 
In the presence of random effects, \cite{entholzner2005popdesign} derive that for single samples the optimal designs for fixed effects models retain their optimality for linear optimality criteria.
\cite{DEBUSHO20081116} show that this also holds for $D$-optimality in linear models when only the intercept is random.
However, in a multi-sample situation \cite{Schmelter2007} and \cite{schmelter2008optimal} exhibit that the variability of the intercept has a non-negligible influence on the $D$-optimal design. 
In the case of random slope effects this dependence already occurs in single samples as outlines by \cite{randomslope2007}.
\cite{dette2010optimal} consider the problem of constructing $D$-optimal designs for linear and  nonlinear random effect models with applications in population
pharmacokinetics. 
These authors present a new approach to determine efficient designs for nonlinear least squares estimation which addresses the problem of additional correlation between observations within units.
Based on geometrical arguments, \cite{grasshoff2012optimal} derive $D$-optimal designs for random coefficient regression models when only one observation is available per unit, a situation which occurs in destructive testing. 
\cite{mentre1997optimal} present an approach to optimal design of experiments for random effects regression models in the presence of cost functions related to costs per unit and costs per measurement with applications to toxicokinetics. 

The present approach is based on the discussion paper by \cite{doi:10.1002/asmb.2061} in which two case studies are introduced for optimal planning of repeated measures Accelerated Degradation Tests. 
There the authors consider the influence of a single stress variable and use a criterion based on a large-sample approximation of the precision for estimating a quantile of the failure-time distribution under normal use conditions. 
We will adopt this approach, generalize the results presented there to more general models, and extend the design optimization also to generate an optimal time plan. 

The present paper is organized as follows. 
Section~\ref{sec-intro-exmpl} starts with a motivation example based on a case study in \citep{doi:10.1002/asmb.2061}. 
In Sections~\ref{sec-model-formulation}, \ref{sec-estimation} and \ref{sec-information} we state the general model formulation, specify the maximum-likelihood estimation and exhibit the corresponding information matrix.
Basic concepts of optimal design theory in the present context are collected in Section~\ref{sec-design} while Section~\ref{sec-failuretime} is devoted to the idea of soft failure due to degradation, where we derive the design optimality criterion for estimating a quantile of the failure time distribution under normal use conditions.
In Section~\ref{sec-od-fixed-tau} optimal designs are characterized when the time plan for repeated measurements at the testing units is fixed in advance.
The paper closes with a short discussion in Section~\ref{sec-discussion}.

\section{Introductory example}
\label{sec-intro-exmpl}
Before formulating our general degradation model in section \ref{sec-model-formulation}, we start in this section for motivation with the description of a simple introductory example based on \citep{doi:10.1002/asmb.2061}.

\begin{example}
	\label{ex-intro}
The model proposed in \cite{doi:10.1002/asmb.2061} is a linear mixed effect model with a single stress variable $x$. 
In this model there are $n$ testing units for which degradation $y_{i j}$ is observed at $k$ time points $t_j$, $j = 1, ..., k$.
The (standardized) stress variable $x$ can be chosen by the experimenter from the design region $\mathcal{X} = [0, 1]$.
On the unit level the response $y_{i j}$ for the degradation of testing unit $i$ at time $t_j$ is represented by
\begin{equation} 
	\label{motivation}
	y_{i j} = 	\beta_{i, 1} + \beta_{2} x_i + \beta_{i, 3} t_j + \beta_{4} x_i t_j + \varepsilon_{i j} ,
\end{equation}
where the intercept $\beta_{i, 1}$ is the mean degradation of unit $i$ at time $t = 0$ under the stress level $x = 0$, $\beta_{2}$ is the common (not unit specific) mean increase in degradation depending on the stress variable $x$, $\beta_{i, 3}$ is the mean increase in degradation of unit $i$ over time $t$ when the stress level is set to $x = 0$, and $\beta_{4}$ is the interaction effect between time and stress.
The measurement errors $\varepsilon_{i j}$ are assumed to be realizations of a normally distributed error variable with mean zero and error variance $\sigma^2_\varepsilon$. 

On the whole experiment level the unit specific parameters $(\beta_{i, 1}, \beta_{i, 3})^T$ of the units are assumed to be realizations of a bivariate normal distribution with mean $(\beta_{1}, \beta_{3})^T$ and a variance covariance matrix $\BS{\Sigma} =
\left(\begin{array}{cc}
	\sigma_1^2 & \rho \sigma_1 \sigma_2
	\\
	\rho \sigma_1 \sigma_2 & \sigma_2^2
\end{array}\right)$.
All random effect parameters and measurement errors are assumed to be independent both within as well as between units.
Under well controlled measuring testing conditions, the variability of the response is completely described by both the unit to unit variability $\BS\Sigma$ and the within unit variability of the measurement errors. 

To illustrate the situation some virtual degradation paths $y_{i j}$, $j = 1, ..., k$, are depicted in Figure~\ref{spaghettiplot} (left panel) for three different values of the stress variable $x$.
There are three units shown at each value of the stress level ($n = 9$) and $k=11$ equally spaced measurement times $t_j$.
The roughness of the paths is due to the measurement errors $\varepsilon_{i j}$
The corresponding underlying mean degradation paths $\mu_i(x_i, t) = \beta_{i, 1} + \beta_{2} x_i + \beta_{i, 3} t + \beta_{4} x_i t$, corrected for the measurement errors, are shown in the right panel of Figure~\ref{spaghettiplot}.
These mean degradation paths are represented by straight lines over time, where both the intercept and the slope may vary across units around an aggregate value determined by the value $x_i$ of the stress variable.
\begin{figure}
	\label{spaghettiplot}
	\centering
		\includegraphics[width=0.4\textwidth]{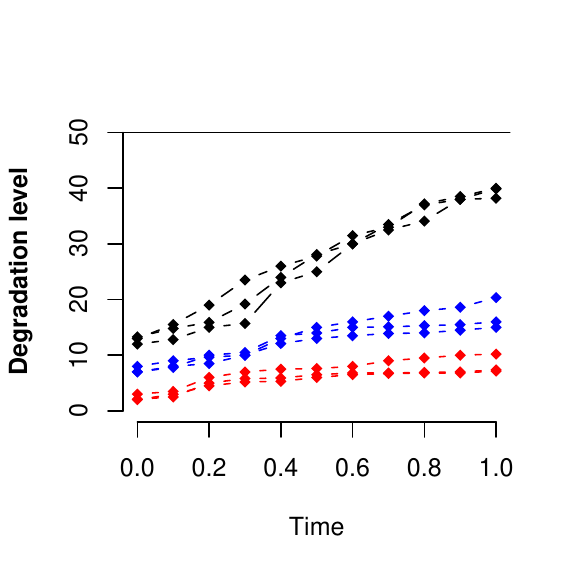}
		\hspace{15mm}
		\includegraphics[width=0.4\textwidth]{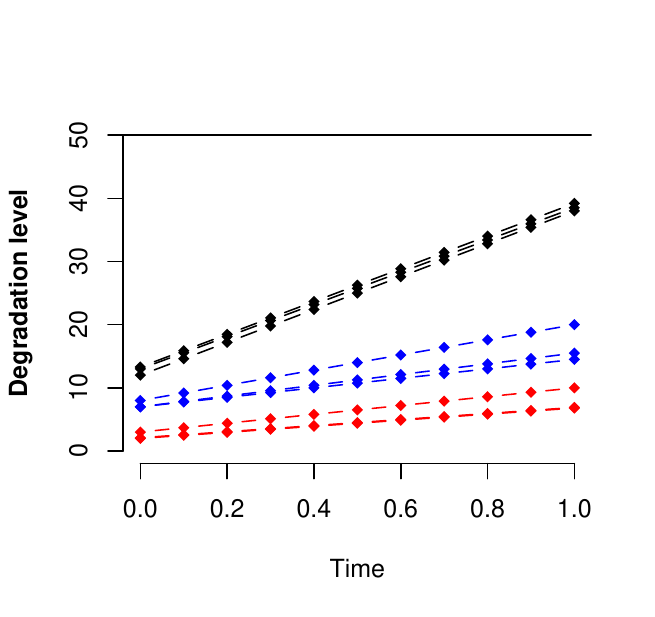}
	\caption{Observed degradation paths (left panel) and corresponding mean degradation paths (right panel)}
\end{figure}
The corresponding time $t_{u}$, for which $\mu_u(x_u, t_{u}) = y_0$, will be called the failure time of unit $u$ under normal use condition due to degradation.
These failure times vary across different unit because of the unit specific parameters $\beta_{u, 1}$ and $\beta_{u, 3}$.

In both panels of Figure~\ref{spaghettiplot} the predetermined failure threshold $y_0 = 50$ is indicated by a horizontal line.
As typical for degradation studies failure does not occur during the time of experiment even for the highest stress level.
\end{example}

\section{Formulation of the model}
\label{sec-model-formulation}

In this section, we give a general formulation of a mixed effects regression model incorporating a product-type structure with complete interactions between the stress and the time variable. 
To become more specific we assume that there are $n$ testing units $i=1,...,n$, for which degradation $y_{i j}$ is to be measured at $k$ subsequent time points $t_j$, $j = 1, ..., k$, $t_1 < ... < t_k$.
Each unit $i$ is observed under a value $\mathbf{x}_i$ of the stress variable(s), which is kept fixed for each unit throughout the  degradation process, but may differ from unit to unit.
The number $k$ of measurements and the time points are the same for all units.  
The measurements $y_{i j}$ are regarded as realizations of random variables $Y_{i j}$ which are described by a hierarchical model. 
For each unit $i$ the observation $Y_{i j}$ at time point $t_j$ is given by
\begin{equation} 
\label{modelindividualresponse}
Y_{i j} = \mu_{i}(\mathbf{x}_i, t) + \varepsilon_{i j} ,
\end{equation}
where $\mu_{i}(\mathbf{x}, t)$ is the mean degradation of unit $i$ at time $t$, when stress $\mathbf{x}$ is applied to unit $i$, and $\varepsilon_{i j}$ is the associated measurement error at time point $t_j$.
The mean degradation $\mu_{i}(\mathbf{x}, t)$ is assumed to be given by a linear model equation in the stress variable $\mathbf{x}$ and time $t$,
\begin{equation} 
	\label{mean-degradation-unit}
	\mu_{i}(\mathbf{x}, t) = \sum_{r = 1}^{p} \beta_{i, r} f_{r}(\mathbf{x}, t_j) 
	= \mathbf{f}(\mathbf{x}, t_j)^{T} \BS{\beta}_i  
\end{equation}
where $\mathbf{f}(\mathbf{x}, t) = (f_{1}(\mathbf{x}, t), ..., f_{p}(\mathbf{x}, t))^T$ is a $p$-dimensional vector of known regression functions $f_{q}(\mathbf{x}, t)$ in both the stress variable(s) $\mathbf{x}$ and the time $t$, $\BS{\beta}_i = (\beta_{i, 1}, ..., \beta_{i, p})^T$ is a $p$-dimensional vector of unit specific parameters $\beta_{i, q}$.
Hence, the response is given by
\begin{equation} 
	\label{modelindividuallevel}
	Y_{i j} = \mathbf{f}(\mathbf{x}_{i}, t_j)^{T}\BS{\beta}_i + \varepsilon_{i j}.
\end{equation}
The measurement error $\varepsilon_{i j}$ is assumed to be normally distributed with zero mean and some potentially time dependent error variance $\sigma_{\epsilon j}^{2}$ ($\varepsilon_{i j} \sim \mathrm{N}(0,\sigma_{\epsilon j}^{2})$).
Moreover, the error terms may be correlated within a unit over time. 
So, in general the vector $\BS{\varepsilon}_{i} = (\varepsilon_{i 1}, ..., \varepsilon_{i k})^T$ of errors associated with the $k$ observations within one unit $i$ is $k$-dimensional multivariate normally distributed with zero mean and positive definite variance covariance matrix $\BS{\Sigma}_\varepsilon$ ($\BS{\varepsilon}_i \sim \mathrm{N}(\mathbf{0}, \BS{\Sigma}_\varepsilon)$).
For the regression functions $\mathbf{f}(\mathbf{x},t)$ we suppose a product-type structure with complete interactions between the stress variable $\mathbf{x}$ and the time $t$, i.\,e.\ there are marginal regression functions $\mathbf{f}_1(\mathbf{x})=(f_{1 1}(\mathbf{x}), ..., f_{1 p_1}(\mathbf{x}))^T$ and $\mathbf{f}_2(t)=(f_{2 1}(t), ..., f_{2 p_2}(t))^T$ of dimension $p_1$ and $p_2$ which only depend on the stress variable $\mathbf{x}$ and the time $t$, respectively, and the vector $\mathbf{f}(\mathbf{x}, t) = \mathbf{f}_1(\mathbf{x}) \otimes \mathbf{f}_2(t)$ of regression functions factorizes into its marginal counterparts ($p = p_1 p_2$).
Here ``$\otimes$'' denotes the Kronecker product of matrices or vectors.
Then the observation $Y_{i j}$ can be written as
\begin{equation} 
	\label{modelindividuallevel-product}
	Y_{i j} = \sum_{r = 1}^{p_1} \sum_{s = 1}^{p_2} \beta_{i, rs} f_{1 r}(\mathbf{x}_{i}) f_{2 s}(t_j) + \varepsilon_{i j}
	= (\mathbf{f}_1(\mathbf{x}) \otimes \mathbf{f}_2(t))^T \BS{\beta}_i + \varepsilon_{i j},
\end{equation}
where for notational convenience the entries of the vector $\BS{\beta}_i = (\beta_{i, 1 1}, ..., \beta_{i, 1 p_2}, ..., \beta_{i, p_1 p_2})$ of parameters are relabeled lexicographically according to their associated marginal regression functions ($q = (r - 1) p_2 + s$, $r = 1, ..., p_1$, $s = 1, ..., p_2$).
Moreover, 
we will assume throughout that the marginal regression function $\mathbf{f}_1(\mathbf{x}) = (f_{11}(\mathbf{x}), ..., f_{1 p_1}(\mathbf{x}))^T$ of the stress variable $\mathbf{x}$ contains a constant term, $f_{1 1}(\mathbf{x}) \equiv 1$ say, which is a common assumption in the majority of situations, and that only the leading $p_2$ parameters $\beta_{i, 1 1}, ..., \beta_{i, 1 p_2}$ of $\BS{\beta}_i$ associated with this constant term are unit specific. 
All other parameters in $\BS{\beta}_i$ are assumed to take the same value $\beta_{rs}$, $r = 2, ..., p_1$, $s = 1, ..., p_2$, for all individuals $i = 1, ..., n$.
Hence, for unit $i$ the model \eqref{modelindividuallevel-product} can be rewritten as
\begin{equation} 
\label{eq-general-model}
 Y_{i j} = \left(\mathbf{f}_1(\mathbf{x}_{i}) \otimes \mathbf{f}_2(t_j)\right)^{T}\BS{\beta} + \mathbf{f}_2(t_j)^T\BS{\gamma}_i + \varepsilon_{i j},
\end{equation}
where $\BS{\beta} = (\beta_{1 1}, ..., \beta_{p_1 p_2})^T$ is the vector of fixed effect (aggregate) parameters (averaged over the units) associated with the constant term in the regression functions of the stress variable $\mathbf{x}$ and $\BS{\gamma}_i = (\gamma_{i 1}, ..., \gamma_{i p_2})^T$ is the $p_2$-dimensional vector of unit specific deviations $\gamma_{i s} = \beta_{i, 1 s} - \beta_{1 s}$, $s = 1, ..., p_2$, from the corresponding aggregate parameters.
On the aggregate level it is assumed that the units are representatives of a larger entity.
The deviations of the units from the aggregate value are then modeled as random effects, i.\,e.\ they are $p_2$-dimensional multivariate normal with zero mean and variance-covariance matrix $\BS{\Sigma}_\gamma$ ($\BS{\gamma}_i \sim \mathrm{N}(\mathbf{0}, \BS{\Sigma}_\gamma)$).
All vectors $\BS{\gamma}_i$ of random effects and all vectors $\BS{\varepsilon}_i$ of measurement errors are assumed to be independent.
In vector notation the $k$-dimensional vector $\mathbf{Y}_{i} = (Y_{i 1},...,Y_{i k})^T$ of observations for unit $i$ can be expressed as 
   \begin{equation*} 
\label{eq-gen-mod-i}
\mathbf{Y}_{i} = (\mathbf{f}_{1}(\mathbf{x}_{i})^T \otimes \mathbf{F}_2) \BS\beta + \mathbf{F}_2\,\BS{\gamma}_i + \BS\varepsilon_{i} ,
\end{equation*}
where $\mathbf{F}_2 = ( \mathbf{f}_2(t_1), ..., \mathbf{f}_2(t_k))^T$ is the $k \times p_2$ marginal design matrix for the time variable. 
Then $\mathbf{Y}_i$ is $k$-dimensional multivariate normally distributed with mean $(\mathbf{f}_{1}(\mathbf{x}_{i})^T \otimes \mathbf{F}_2) \BS\beta$ and variance covariance matrix $\mathbf{V}=\mathbf{F}_2{\BS{\Sigma}}_\gamma\mathbf{F}_2^T + {\BS{\Sigma}}_{\varepsilon}$.
The variance covariance matrix $\mathbf{V}$ is not affected by the choice of the stress level $\mathbf{x}_i$ and, hence, equal for all units $i$.
In total, for the observations of all $n$ units the stacked $n k$-dimensional response vector $\mathbf{Y}=(\mathbf{Y}_{1}^T, ... , \mathbf{Y}_{n}^T)^T$ can be represented in matrix notation as
\begin{equation}
   	\label{full-model} 
	\mathbf{Y} = ({\mathbf{F}_1 \otimes \mathbf{F}_2})\BS\beta + (\mathbf{I}_n \otimes {\mathbf{F}_2})\,{\BS\gamma} + \BS\varepsilon ,
\end{equation} 
where $ \mathbf{F}_1 = (\mathbf{f}_1(\mathbf{x}_1), ..., \mathbf{f}_1(\mathbf{x}_n))^{T}$ is the $n \times p_1$ marginal design matrix for the stress variables across units, ${\BS\gamma} = ({\BS\gamma}_1^T, ..., {\BS\gamma}_n^T )^T$ is the $n p_2$-dimensional stacked parameter vector of random effects and ${\BS\varepsilon} = ({\BS\varepsilon}_1^T, ..., {\BS\varepsilon}_n^T )^T$ is the $n k$-dimensional stacked vector of random errors.
Such a model equation is sometimes called the ``marginal model'' for the response $\mathbf{Y}$, but should not be confused with models marginalized for the covariates $\mathbf{x}$ and $t$, respectively (see the decomposition at the end of Section~\ref{sec-information}). 
Note that the vectors ${\BS\gamma}\sim \mathrm{N}(\mathbf{0},\mathbf{I}_n\otimes {\BS{\Sigma}}_\gamma)$ of all random effects and the vector $\BS \varepsilon \sim \mathrm{N}(\mathbf{0},  \mathbf{I}_n\otimes {\BS{\Sigma}}_{\varepsilon})$ are multivariate normal.
Hence, the vector $\mathbf{Y}$ of all observations is $n k$-dimensional multivariate normal, $\mathbf{Y}\sim \mathrm{N}(\mathbf{0}, \mathbf{I}_n \otimes \mathbf{V})$.
For the analysis of degradation under normal use we further assume that the general model~\ref{eq-general-model} is also valid at the normal use condition $\mathbf{x}_u$, where typically $\mathbf{x}_u \not\in \mathcal{X}$, i.\,e.\ 
\begin{equation}
	\label{eq-degr-path-use-cond}
	\mu(\mathbf{x}_{u},t) = (\mathbf{f}_1(\mathbf{x}_{u}) \otimes \mathbf{f}_2(t))^{T} \BS{\beta} + \mathbf{f}_2(t)^T \BS{\gamma}_u
\end{equation} 
describes the mean degradation of a future unit $u$ at normal use condition $\mathbf{x}_u$ and time $t$, and the random effects $\BS{\gamma}_u$ are $p_2$-dimensional multivariate normal with mean zero and variance covariance matrix $\BS\Sigma_\gamma$.

\section{Estimation of the model parameters}
\label{sec-estimation}
Under the distributional assumptions of normality for both the random effects and the measurement errors the model parameters may be estimated by means of the maximum likelihood method.
Denote by $\BS\theta = (\BS{\beta}^T, \BS{\varsigma}^T)^T$ the vector of all model parameters, where $\BS\varsigma$ collects all variance covariance parameters from $\BS\Sigma_\gamma$ and $\BS\Sigma_\varepsilon$
For the general model~(\ref{full-model}) the log-likelihood is given by
\begin{equation}
	\label{eq-loglikelihood}
	\ell(\BS{\theta}; \mathbf{y}) = 
	- {\textstyle{\frac{n k}{2}}} \log(2\pi) 
	- {\textstyle{\frac{n}{2}}} \log(\det(\mathbf{V})) 
	- {\textstyle{\frac{1}{2}}} (\mathbf y - (\mathbf{F}_{1} \otimes \mathbf{F}_2) \BS{\beta})^T (\mathbf{I}_n \otimes \mathbf{V})^{-1} (\mathbf{y} - (\mathbf{F}_{1} \otimes \mathbf{F}_2) \BS{\beta}) ,
\end{equation}
where the variance covariance matrix $\mathbf{V} = \mathbf{V}(\BS\varsigma)$ of measurements per unit depends only on $\BS\varsigma$.
The maximum likelihood estimator of $\BS\beta$ can be calculated as
\begin{eqnarray} 
	\widehat{\BS{\beta}} 
	&=&
	((\mathbf{F}_1 \otimes \mathbf{F}_2)^T (\mathbf{I}_n \otimes \widehat{\mathbf{V}})^{-1} (\mathbf{F}_1 \otimes \mathbf{F}_2))^{-1} (\mathbf{F}_1 \otimes \mathbf{F}_2)^T (\mathbf{I}_n \otimes \widehat{\mathbf{V}})^{-1} \mathbf{Y}
	\nonumber
	\\
	&=&
	((\mathbf{F}_1^T \mathbf{F}_1)^{-1} \mathbf{F}_1^T) \otimes 
	((\mathbf{F}_2^T \widehat{\mathbf{V}}^{-1} \mathbf{F}_2)^T \mathbf{F}_2^T \widehat{\mathbf{V}}^{-1}) \mathbf{Y} ,
	\label{eq-maxlik-beta}
\end{eqnarray}
if both $\mathbf{F}_1$ and $\mathbf{F}_2$ are of full column rank $p_1$ and $p_2$, respectively, and $\widehat{\mathbf{V}}=\mathbf{V}(\widehat{\BS\varsigma})$, where $\widehat{\BS\varsigma}$ is the maximum likelihood estimator of $\BS\varsigma$.
When $\mathbf{V}$ is known, at least up to a multiplicative constant, $\mathbf{V} = \sigma^2 \mathbf{V}_0$, then $\widehat{\BS\beta}$ is the best liner unbiased (general least squares) estimator $\widehat{\BS{\beta}}_\mathrm{GLS} = ((\mathbf{F}_1^T \mathbf{F}_1)^{-1} \mathbf{F}_1^T) \otimes ((\mathbf{F}_2^T \mathbf{V}_0^{-1} \mathbf{F}_2)^T \mathbf{F}_2^T \mathbf{V}_0^{-1}) \mathbf{Y}$ of $\BS\beta$.
In particular, when the measurement errors are uncorrelated and homoscedastic, i.\,e.\ $\BS\Sigma_\varepsilon = \sigma_\varepsilon^2 \mathbf{I}_k$, then this estimator reduces to the ordinary least squares estimator $\widehat{\BS{\beta}}_\mathrm{OLS} =  ((\mathbf{F}_1^T \mathbf{F}_1)^{-1} \mathbf{F}_1^T) \otimes ((\mathbf{F}_2^T \mathbf{F}_2)^T \mathbf{F}_2^T) \mathbf{Y}$  by a result of \cite{zyskind1967blue-ols} because $\mathbf{V} \mathbf{F}_2 = \mathbf{F}_2 (\BS\Sigma_{\gamma} \mathbf{F}_2^T \mathbf{F}_2 + \sigma_\varepsilon^2 \mathbf{I}_{p_2})$.
Hence, in the case of uncorrelated homoscedastic measurement errors the maximum likelihood estimator of the location parameters $\BS\beta$ does neither depend on the variance covariance parameters nor on their estimates.
In general, the quality of the estimator $\widehat{\BS{\beta}}$ can be measured in terms of its variance covariance matrix which is given by
\begin{equation} 
	\label{eq-covariance-estimator}
	\mathrm{Cov}(\widehat{\BS{\beta}}) = (\mathbf{F}_1^T \mathbf{F}_1)^{-1} \otimes (\mathbf{F}_2^T \mathbf{V}^{-1} \mathbf{F}_2)^{-1}.
\end{equation}
By using the structure $\mathbf{V} = \mathbf{F}_2 \BS\Sigma_\gamma \mathbf{F}_2^T + \BS\Sigma_{\varepsilon}$ the last term can be calculated as
\begin{equation}
	\label{cov_f2_calc}
	(\mathbf{F}_2^T \mathbf{V}^{-1} \mathbf{F}_2)^{-1} = (\mathbf{F}_2^T \BS\Sigma_{\varepsilon}^{-1} \mathbf{F}_2)^{-1} + \BS\Sigma_\gamma 
\end{equation}
in terms of the variance covariance matrices $\BS\Sigma_\gamma$ and $\BS\Sigma_{\varepsilon}$ of the random effects and the measurement errors, respectively.

\section{Information}  
\label{sec-information}
In general, the Fisher information matrix is defined as the variance covariance matrix of the score function $\mathbf{U}$ which itself is defined as the vector of first derivatives of the log likelihood with respect to the components of the parameter vector $\BS{\theta}$.
More precisely, let $\mathbf{U} = \left(\frac{\partial}{\partial \theta_1} \ell(\BS\theta; \mathbf{Y}), ..., \frac{\partial}{\partial \theta_q} \ell(\BS\theta; \mathbf{Y})\right)^T$, where $q$ is the dimension of $\BS\theta$.
Then for the full parameter vector $\BS\theta$ the Fisher information matrix is defined as $\mathbf{M}_{\BS\theta} = \mathrm{Cov}(\mathbf{U})$, where the expectation is taken with respect to the distribution of $\mathbf{Y}$.
The Fisher information can also be computed as minus the expectations of the second derivatives of the score function $U$, i.\,e.\ $\mathbf{M}_{\BS\theta} = - \mathrm{E}\left(\frac{\partial^2}{\partial \BS\theta \partial \BS\theta^T} \ell(\BS\theta; \mathbf{Y})\right)$. 
Under common regularity conditions the maximum likelihood estimator $\widehat{\BS\theta}$ of $\BS\theta$ is consistent and asymptotically normal with asymptotic variance covariance matrix equal to the inverse $\mathbf{M}_{\BS\theta}^{-1}$ of the Fisher information matrix $\mathbf{M}_{\BS\theta}$.
To specify the Fisher information matrix further, denote by $\mathbf{{M}}_{\BS\beta} = - \mathrm{E}\left(\frac{\partial^2}{\partial \BS\beta \partial \BS\beta^T} \ell(\BS\theta; \mathbf{Y})\right)$, $\mathbf{M}_{\BS\varsigma} = - \mathrm{E}\left(\frac{\partial^2}{\partial \BS\varsigma \partial \BS\varsigma^T} \ell(\BS\theta; \mathbf{Y})\right)$, $\mathbf{M}_{\BS\beta \BS\varsigma} = - \mathrm{E}\left(\frac{\partial^2}{\partial \BS\beta \partial \BS\varsigma^T} \ell(\BS\theta; \mathbf{Y})\right)$ and $\mathbf{{M}}_{\BS\varsigma \BS\beta} = \mathbf{{M}}_{\BS\beta \BS\varsigma}^T$ the blocks of the Fisher information matrix corresponding to the second derivatives with respect to $\BS\beta$ and $\BS\varsigma$ and the mixed derivatives, respectively. 
The mixed blocks can be seen to be zero and the Fisher information matrix is block diagonal,
\begin{equation}
	\label{info-block}
	\mathbf{M}_{\BS{\theta}}
	= \left( 
		\begin{array}{cc}
			\mathbf{M}_{\BS\beta} & \mathbf{0}
			\\
			\mathbf{0} & \mathbf{M}_{\BS\varsigma}
		\end{array}
	\right) .
\end{equation}  
Moreover, the block $\mathbf{M}_{\BS\beta}$ associated with the aggregate location parameters $\BS\beta$ can be determined as
\begin{equation}
	\mathbf{M}_{\BS\beta} = (\mathbf{F}_{1}^T \mathbf{F}_{1}) \otimes (\mathbf{F}_2^{T} \mathbf V^{-1}  \mathbf{F}_2) 
\end{equation}
which turns out to be the inverse of the variance covariance matrix for the estimator $\widehat{\BS\beta}$ of $\BS\beta$, when $\mathbf{V}$ is known.
Actually, because the Fisher information matrix for $\BS\theta$ is block diagonal, the inverse $	\mathbf{M}_{\BS\beta}^{-1} = (\mathbf{F}_1^T \mathbf{F}_1)^{-1} \otimes (\mathbf{F}_2^T \mathbf{V}^{-1} \mathbf{F}_2)^{-1}$ of the block associated with $\BS\beta$ is the corresponding block of the inverse of $\mathbf{M}_{\BS\theta}$ and is, hence, the asymptotic variance covariance matrix of $\widehat{\BS{\beta}}$.

Accordingly the asymptotic variance covariance matrix for estimating the variance parameters $\BS\varsigma$ is the inverse of the block $\mathbf{M}_{\BS\varsigma}$.
In the following we will call $\mathbf{M}_{\BS\beta}$ and $\mathbf{M}_{\BS\varsigma}$ the information matrices for $\BS\beta$ and $\BS\varsigma$, respectively, for short.
The particular form of $\mathbf{M}_{\BS\varsigma}$ will be not of interest here.
However, as the information matrix $\mathbf{M}_{\BS\varsigma}$ for the variance parameters $\BS\varsigma$ is given by
\begin{equation*}
	\label{12345678934tg}
	\textbf{M}_{\BS\varsigma} = \frac{n}{2} \left(\frac{\partial^2 \log(\det(\mathbf{V}))}{\partial \BS\varsigma \partial \BS\varsigma^T} + \mathrm{tr}\left(\mathbf{V} \frac{\partial^2 \mathbf{V}^{-1}}{\partial \BS\varsigma \partial \BS\varsigma^T}\right)\right).
\end{equation*}
It is important to note that $\mathbf{M}_{\BS\varsigma}$ does not depend on the settings $\mathbf{x}_1, ..., \mathbf{x}_n$ of the stress variable in contrast to the information matrix $\mathbf{M}_{\BS\beta}$ of the aggregate location parameters $\BS\beta$.
For the general product-type model~(\ref{full-model}) the information matrix $\mathbf{M}_{\BS\beta}$ for the aggregate parameters $\BS\beta$ factorizes according to
\begin{equation}
	\label{eq-info-product}
	\mathbf{M}_{\BS\beta} = \mathbf{M}_1 \otimes \mathbf{M}_2
\end{equation}
into the information matrix $\mathbf{M}_1 = \mathbf{F}_1^T \mathbf{F}_1$ in the marginal model 
\begin{equation}
	\label{eq-marginal-1}
	Y_{i}^{(1)} = \mathbf{f}_1(\mathbf{x}_i)^T \BS\beta^{(1)} + \varepsilon_i^{(1)} ,
\end{equation}
$i=1,...,n$, in the stress variable $\mathbf{x}$ with standardized uncorrelated homoscedastic error terms, $\sigma^2_{\varepsilon^{(1)}} = 1$, and the information matrix $\mathbf{M}_2 = \mathbf{F}_2^T\mathbf{V}^{-1}\mathbf{F}_2$ in the mixed effects marginal model 
\begin{equation}
	\label{eq-marginal-2}
	Y_{j}^{(2)} = \mathbf{f}_2(t_j)^T \BS\beta^{(2)} + \mathbf{f}_2(t_j)^T \BS\gamma^{(2)} + \varepsilon_j^{(2)} ,
\end{equation}
$j=1,...,k$, in the time variable $t$ with variance covariance matrices $\BS\Sigma_{\gamma}$ and $\BS\Sigma_{\varepsilon}$ for the random effects $\BS\gamma^{(2)}$ and measurement errors $\BS\varepsilon^{(2)} = (\varepsilon_1^{(2)}, ..., \varepsilon_k^{(2)})^T$, respectively.
Then the information matrix $\mathbf{M}_{\BS\theta}$ in the full model depends on the settings $\mathbf{x}_1, ..., \mathbf{x}_n$ of the stress variable only through the information matrix $\mathbf{M}_1$ in the first marginal model.

\section{Design}
\label{sec-design} 
The quality of the estimates will be measured in terms of the information matrix and, hence, depends on both the settings of the stress variable and the time points of measurements. When these variables are under the control of the experimenter, then their choice will be called the design of the experiment.
Here we assume that the time plan $\mathbf{t} = (t_1, ..., t_k)^T$ for the time points of measurements within units is fixed in advance and is not under disposition of the experimenter.
Then only the settings $\mathbf{x}_1, ..., \mathbf{x}_n$ of the stress variable $\mathbf{x}$ can be adjusted to the units $i = 1, ..., n$.
Their choice $(\mathbf{x}_1, ..., \mathbf{x}_n)$ is then called an ``exact'' design, and their influence on the performance of the experiment is indicated by adding them as an argument to the information matrices, $\mathbf{M}_{\BS\theta}(\mathbf{x}_1, ..., \mathbf{x}_n)$, $\mathbf{M}_{\BS\beta}(\mathbf{x}_1, ..., \mathbf{x}_n)$, and $\mathbf{M}_{1}(\mathbf{x}_1, ..., \mathbf{x}_n)$, where appropriate.
Remind that both $\mathbf{M}_{\BS\varsigma}$ and $\mathbf{M}_{2}$ do not depend on the design for the stress variable.

As $\mathbf{M}_{1}(\mathbf{x}_1, ..., \mathbf{x}_n) = \sum_{i = 1}^n \mathbf{f}(\mathbf{x}_i) \mathbf{f}(\mathbf{x}_i)^T$ it can easily be seen that the information matrices do not depend on the order of the setting but only on their mutually distinct settings, $\mathbf{x}_1, ..., \mathbf{x}_m$ say, and their corresponding frequencies $n_1, ..., n_m$, such that $\sum_{i = 1}^m n_i = n$, i.\,e.\ $\mathbf{M}_{1} = \sum_{i = 1}^m n_i \mathbf{f}(\mathbf{x}_i) \mathbf{f}(\mathbf{x}_i)^T$.
Finding optimal exact designs is, in general, a difficult task of discrete optimization.
To circumvent this problem we follow the approach of approximate designs propagated by \cite{kiefer1959optimum} in which the requirement of integer numbers $n_i$ of testing units at a stress level $\mathbf{x}_i$ is relaxed.
Then continuous methods of convex optimization can be employed (see e.\,g.\ \cite{silvey1980optimal}) and efficient exact designs can be derived by rounding the optimal numbers to nearest integers.
This approach is, in particular, of use when the number $n$ of units is sufficiently large.
Moreover, the frequencies $n_i$ will be replaced by proportions $w_i=n_i/n$, because the total number $n$ of units does not play a role in the optimization. 
Thus an approximate design $\xi$ is defined by a finite number of settings $\mathbf{x}_i$, $i=1,...,m$, from the experimental region $\mathcal{X}$ with corresponding weights $w_i\geq 0$ satisfying $\sum_{i = 1}^m w_i = 1$ and is denoted by
\begin{equation}
	\label{eq-design-x}
	\xi = \left(
		\begin{array}{ccc}
			\mathbf{x}_1& ... &\mathbf{x}_{m} 
			\\ 
			w_{1}& ... &w_{m}
		\end{array} 
	\right),
\end{equation}\\
The corresponding standardized, per unit information matrices are accordingly defined as
\begin{equation}
	\label{eq-info_1-xi}
	\mathbf{M}_1(\xi)
	= \sum_{i = 1}^m w_i \mathbf{f}_{1}(\mathbf{x}_{i}) \mathbf{f}_{1}(\mathbf{x}_{i})^T 
\end{equation}
for the marginal model on itself or by plugging (\ref{eq-info_1-xi}) in into the standardized, per unit information matrix
\begin{equation}
	\label{eq-info_beta-xi}
	\mathbf{M}_{\BS\beta}(\xi)
	= \mathbf{M}_{1}(\xi) \otimes \mathbf M_2
\end{equation}
for the aggregate parameters $\BS\beta$, where again $\mathbf{M}_2=\mathbf{F}_2^{T} \mathbf V^{-1} \mathbf{F}_2$, and
\begin{equation}
	\label{eq-info_theta-xi}
	\mathbf{M}_{\BS\theta}(\xi)
	= \left(
		\begin{array}{cc}
			 \mathbf{M}_{1}(\xi) \otimes \mathbf M_2 & \mathbf{0}\\
			 \mathbf{0} & \widetilde{\mathbf{M}}_{\BS\varsigma}
		\end{array}
	\right)
\end{equation}
or the full parameter vector $\BS\theta$, where now $\widetilde{\mathbf{M}}_{\BS\varsigma} = \frac{1}{n} \mathbf{M}_{\BS\varsigma}$ is the standardized, per unit information for the variance parameters $\BS\varsigma$.
If all $n w_i$ are integer, then these standardized versions coincide with the information matrices of the corresponding exact design up to the normalizing factor $1/n$ and are, hence, an adequate generalization.
In order to optimize information matrices, some optimality criterion has to be employed which is a real valued function of the information matrix and reflects the main interest in the experiment.
\section{Optimality criterion based on the failure time under normal use condition}
\label{sec-failuretime}
As in \cite{doi:10.1002/asmb.2061} we are interested in some characteristics of the failure time distribution of soft failure due to degradation. 
Therefore it is assumed that the model equation (\ref{eq-degr-path-use-cond}) $\mu_{u}(t) = \mu(\mathbf{x}_{u}, t) = (\mathbf{f}_1(\mathbf{x}_{u}) \otimes \mathbf{f}_2(t))^{T} \BS{\beta} + \mathbf{f}_2(t)^T \BS{\gamma}_u$ for the mean degradation paths is also valid under normal use condition $\mathbf{x}_u$, where $\mu_u$ denotes the degradation path under normal use condition for short.
We further denote by $\mu(t) = \mathrm{E}(\mu_u(t)) = (\mathbf{f}_1(\mathbf{x}_{u}) \otimes \mathbf{f}_2(t))^{T}\BS{\beta}$ the aggregate degradation path under normal use condition and by $\BS\delta = \BS\delta(\BS\beta) = (\delta_1(\BS\beta), ..., \delta_{p_2}(\BS\beta))^T$ the vector of its coefficients $\delta_{s} = \delta_{s}(\BS\beta) = \sum_{r = 1}^{p_1} f_{1 r}(\mathbf{x}_u)\beta_{r s}$, $s = 1, ...., p_2$, in the regression functions $f_{2 s}$ in $t$, i.\,e.\ $\mu(t) = \mathbf{f}_2(t)^T\BS\delta = \sum_{s = 1}^{p_2}\delta_s f_{2 s}(t)$.

For the following it is assumed that the mean degradation paths are strictly increasing over time.
Then a soft failure due to degradation is defined as the exceedance of the degradation over a  failure threshold $y_0$.
This definition is based on the mean degradation path and not on a ``real'' path subject to measurement errors.
The failure time $T$ under normal use condition is then defined as the first time $t$ the mean degradation path $\mu_u(t)$ reaches or exceeds the threshold $y_0$, i.\,e.\ $T = \min \{t \geq 0;\, \mu_u(t) \geq y_0\}$.
As random effects $\BS\gamma_u$ are involved in the mean degradation path, the failure time $T$ is random.
Actually, $T$ may become infinite, if the mean degradation path does not reach the threshold, or may degenerate to $T = 0$, if the degradation already exceeds the threshold at time $t=0$, because of unfortunate values of the random effects $\BS\gamma_u$, but this will happen only with low probability and will not affect the further argumentation.

In order to describe certain characteristics of the distribution of the failure time $T$, we will determine its distribution function $F_T(t) = \mathrm{P}(T \leq t)$.
First note that $T \leq t$ if and only if $\mu_u(t) \geq y_0$.
Hence
\begin{eqnarray}
   F_{T}(t) &  = & \mathrm{P}(\mu_u(t) \geq y_0)
   \nonumber
   \\
   & = & \mathrm{P}(\mu(t) + \mathbf{f}_2(t)^T \BS{\gamma}_u \geq y_0) 
   \nonumber
	\\
	& = & \mathrm{P} (- \mathbf{f}_2(t)^T \BS{\gamma}_u \leq \mu(t) - y_0) 
   \nonumber
   \\
	& = & \Phi(h(t)) ,
    \label{eq-failure-time-distribution}
\end{eqnarray}
where
\begin{equation}
	\label{eq-h-tau}
	h(t) = \frac{\mu(t) - y_0}{\sigma_u(t)} ,
\end{equation}
$\sigma_u^2(t) = \mathbf{f}_2(t)^T \BS\Sigma_\gamma \mathbf{f}_2(t)$ is the variance of the mean degradation path $\mu_u(t)$ at time $t$, and $\Phi$ denotes the distribution function of the standard normal distribution.
Here it is tacitly assumed that the variance $\sigma_\mu^2(t)$ of the mean degradation path  is greater than zero for every $t \geq 0$.
This condition is satisfied, in particular, when the variance covariance matrix $\BS\Sigma_\gamma$ of the random effects is positive definite.

We will be interested in quantiles $t_\alpha$ of the failure time distribution, i.\,e.\ $\mathrm{P}(T \leq t_\alpha) = \alpha$.
For each $\alpha$ the quantile $t_\alpha$ gives the time up to which under normal use conditions (at least) $\alpha \cdot 100$ percent of the units fail and (at least) $(1 - \alpha) \cdot 100$ percent of the units persist.
The quantiles $t_\alpha$ are increasing in $\alpha$.
Note that this standard definition of quantiles is in contrast to the``upper'' quantiles ($t_{1-\alpha}$) used in  \cite{doi:10.1002/asmb.2061} where percentages of failures and persistence are reversed.
Of particular importance is the median $t_{0.5}$ up to which under normal use conditions half of the units fails and half of the units persist ($\alpha = 0.5$).
Other characteristics of interest may be the five or ten percent quantiles $t_{0.05}$ and $t_{0.1}$ which give the times up to which $95$ or $90\,\%$ percent of the units persist, respectively.
By (\ref{eq-failure-time-distribution}) these quantiles can be determined as the solutions of the equation 
\begin{equation}
	\label{eq-tau-alpha-implicit}
	h(t_\alpha) = z_\alpha ,
\end{equation}
where $z_\alpha = \Phi^{-1}(\alpha)$ is the $\alpha$-quantile of the standard normal distribution.
For the median ($\alpha = 1/2$) we have $z_{0.5} = 0$ and, hence the median failure time $t_{0.5}$ is the solution of $\mu(t) = y_0$, i.\,e.\ the aggregate degradation path reaches the threshold at time $t_{0.5}$.
Note that the function $h$ represents the failure time distribution function $F_T$ on a normal Q-Q-plot scale. In the particular case of straight lines for the mean degradation paths, i.\,e.\ $\mathbf{f}_2(t) = (1,t)^T,$ 
the function $h(t)$ specifies to
\begin{equation}
	\label{eq-h-tau-linear}
	h(t) = \frac{\delta_2 t + \delta_1 - y_0}{\sqrt{\sigma_1^2 + 2\rho \sigma_1 \sigma_2 t + \sigma_2^2 t^2}} ,
\end{equation}
where $\delta_{1} = \sum_{r = 1}^{p_1} f_{1 r}(\mathbf{x}_u) \beta_{r 1}$ and $\delta_2 = \sum_{r = 1}^{p_1} f_{1 r}(\mathbf{x}_u) \beta_{r 2}$ are the intercept and the slope of the aggregate degradation path $\mu(t) = \delta_1 + \delta_2 t$ under normal use condition, respectively.
The median failure time is then given by $t_{0.5} = (y_0 - \delta_1) / \delta_2$ which provides a proper solution $t_{0.5} > 0$ under the natural assumptions that the aggregate degradation path is increasing, $\delta_2 > 0$, and that the aggregate degradation at the beginning of the testing at time $t=0$ is less than the threshold of soft failure, $\delta_1 < y_0$.


Under the additional assumption that the correlation of the random effects is non-negative for the intercept and the slope of the mean degradation path,  $\rho \geq 0$, the function $h(t)$ can be seen to be strictly increasing, $h^\prime(t) > 0$, in $t>0$.
This also remains true for small to moderate negative correlations.
However, the range of $h(t)$ is bounded and does not cover the whole real line such that not all quantiles are non-degenerate.
For small $\alpha$ the $\alpha$-quantile to be positive requires $z_\alpha > h(0) = - (y_0 - \delta_1)/\sigma_1$, i.\,e.\ the variance $\sigma_1^2$ of the intercept of the mean degradation path has to be sufficiently small compared to the distance from its mean $\delta_1$ to the threshold $y_0$.
In particular, in the case of the $5\,\%$-quantile $\sigma_1<0.608(y_0-\delta_1)$ is needed for $t_{0.05}>0$.
For large $\alpha$ the $\alpha$-quantile is finite if $z_\alpha < \lim_{t\to\infty}h(t) = \delta_2/\sigma_2$, i.\,e.\ the variance $\sigma_2^2$ of the slope of the mean degradation path has to be sufficiently small compared to its mean $\delta_2$. 
Note that $\Phi(h(0)) = 1 - \Phi((y_0 - \delta_1) / \sigma_1)$ is the probability that under normal use condition the mean degradation path exceeds the threshold $y_0$ already at the initial time $t=0$.
Note also that formally $1 - \Phi(\lim_{t \to \infty} h(t)) = \Phi( - \delta_2 / \sigma_2)$ is the probability that the mean degradation path has a negative slope which may be interpreted as the probability that soft failure due to degradation will not occur at all under normal use condition.
When the $\alpha$-quantile is non-degenerate ($0 < t_\alpha < \infty$), then $t_\alpha$ is a solution of the quadratic equation 
\begin{equation*}
	\label{eq-h-tau-linear-solve}
	(\delta_2 t + \delta_1 - y_0)^2 =  z_\alpha^2(\sigma_1^2 + 2\rho\sigma_1\sigma_2 t + \sigma_2^2 t^2) ,
\end{equation*}
as indicated by \cite{doi:10.1002/asmb.2061}.
In the special case of only a random intercept in the random effects, i.\,e.\ $\sigma_2^2 = 0$, all $\alpha$-quantiles $t_\alpha$ finitely exist for $\alpha \geq \Phi( - (y_0-\delta_1) / \sigma_1)$ and can be determined as the solution of a linear equation to $t_\alpha = (y_0 - \delta_1 + z_\alpha\sigma_1) / \delta_2$.

In any case the quantile $t_\alpha = t_\alpha(\BS\theta)$ is a function of both the aggregate location parameters $\BS\beta$ and the variance parameters $\BS\varsigma$, in general.
Hence, the maximum likelihood estimator of the quantile $t_\alpha$ is given by $\widehat{t}_\alpha = t_\alpha(\widehat{\BS\theta})$ in terms of the maximum likelihood estimator $\widehat{\BS\theta}$ of $\BS\theta$.
The task of designing the experiment will now be to provide an as precise estimate of the $\alpha$-quantile as possible.

By the delta-method $\widehat{t}_\alpha$ is seen to be asymptotically normal with  asymptotic variance
\begin{equation}
	\label{eq-avar-tau-alpha}
	\mathrm{aVar}(\widehat{t}_{\alpha}) = \mathbf{c}^T \mathbf{M}_{\BS{\theta}}^{-1} \mathbf{c} ,
\end{equation}
where $\mathbf{c} = \frac{\partial}{\partial\BS\theta}t_{\alpha}$ is the gradient vector of partial derivatives of $t_{\alpha}$ with respect to the components of the parameter vector ${\BS\theta}$. 
The asymptotic variance depends on the design of the experiment through the information matrix $ \mathbf{M}_{\BS{\theta}}$ and will be chosen as the optimality criterion for the design.

The gradient $\mathbf{c}$ can be seen to be equal to
\begin{equation}
	\label{eq-gradient-tau-alpha}
	\mathbf{c} = - c_0 \left({\textstyle{\frac{\partial}{\partial\BS\theta}}}\mu(t)\vert_{t=t_\alpha} - z_\alpha {\textstyle{\frac{\partial}{\partial\BS\theta}}}\sigma_u(t)\vert_{t=t_\alpha}\right) ,
\end{equation}
in view of (\ref{eq-h-tau}) and (\ref{eq-tau-alpha-implicit}) by the implicit function theorem (see e.\,g.\ \cite{krantz2012implicit}), where $c_0 = 1 / (\mu^\prime(t_\alpha) - z_\alpha \sigma_u^\prime(t_\alpha))$ is the inverse of the derivative of the defining function $h$ with respect to $t$.

As the aggregate mean degradation $\mu(t)$ only depends on the aggregate location parameters $\BS\beta$ and the variance $\sigma_u^2(t)$ only depends on the variance parameters $\BS\varsigma$ the gradient simplifies to $\mathbf{c} = - c_0 (\mathbf{c}_{\BS\beta}^T, \mathbf{c}_{\BS\varsigma}^T)^T$, where
\begin{equation*}
	\label{eq-gradient-tau-alpha-beta}
	\mathbf{c}_{\BS\beta} = {\textstyle{\frac{\partial}{\partial\BS\beta}}}\mu(t)\vert_{t = t_\alpha} = \mathbf{f}(\mathbf{x}_u, t_\alpha)
\end{equation*}
is the gradient of $\mu(t)$ with respect to $\BS\beta$ and
\begin{equation*}
	\label{eq-gradient-tau-alpha-sigma}
	\mathbf{c}_{\BS\varsigma} = - z_\alpha {\textstyle{\frac{\partial}{\partial\BS\varsigma}}}\sigma_u(t)\vert_{t = t_\alpha}
\end{equation*}
is $ - z_\alpha$ times the gradient of $\sigma_u(t)$ with respect to $\BS\varsigma$.
The particular shape of $\mathbf{c}_{\BS\varsigma}$ does not play a role here, in general.
But note that $\mathbf{c}_{\BS\varsigma} = \mathbf{0}$ in the case of the median ($\alpha=0.5$).

By the block diagonal form (\ref{info-block}) of the information matrix the asymptotic variance (\ref{eq-avar-tau-alpha}) of $\widehat{t}_\alpha$ becomes
\begin{equation}
	\label{eq-avar-tau-alpha-sum}
	\mathrm{aVar}(\widehat{t}_{\alpha}) = c_0^2 \left(\mathbf{c}_{\BS\beta}^T \mathbf{M}_{\BS{\beta}}^{-1} \mathbf{c}_{\BS\beta} + \mathbf{c}_{\BS\varsigma}^T \mathbf{M}_{\BS{\varsigma}}^{-1} \mathbf{c}_{\BS\varsigma}\right) 
\end{equation}
which simplifies to
\begin{equation}
	\label{eq-avar-tau-median}
	\mathrm{aVar}(\widehat{t}_{0.5}) = c_0^2 \mathbf{c}_{\BS\beta}^T \mathbf{M}_{\BS{\beta}}^{-1} \mathbf{c}_{\BS\beta} 
\end{equation}
in the case of the median.

For the product-type model~(\ref{full-model}) the expression related to the aggregate parameters $\BS\beta$ further decomposes,
\begin{equation}
	\label{eq-avar-tau-alpha-decompose}
	\mathbf{c}_{\BS\beta}^T \mathbf{M}_{\BS{\beta}}^{-1} \mathbf{c}_{\BS\beta} = \mathbf{f}_1(\mathbf{x}_u)^T \mathbf{M}_{1}^{-1} \mathbf{f}_1(\mathbf{x}_u) \cdot \mathbf{f}_2(t_\alpha)^T \mathbf{M}_{2}^{-1} \mathbf{f}_2(t_\alpha) ,
\end{equation}
as $\mathbf{c}_{\BS\beta} = \mathbf{f}_1(\mathbf{x}_u) \otimes \mathbf{f}_2(t_\alpha)$ factorizes.

\section{Optimal designs with predetermined measurement times}
\label{sec-od-fixed-tau}

From (\ref{eq-avar-tau-alpha-sum}) and (\ref{eq-avar-tau-alpha-decompose}) it can be seen that for obtaining a minimal asymptotic variance for $\widehat{t}_\alpha$ only $\mathbf{f}_1(\mathbf{x}_u)^T \mathbf{M}_{1}^{-1} \mathbf{f}_1(\mathbf{x}_u)$ has to be minimized, because all other terms do not depend on the experimental settings $\mathbf{x}_1, ..., \mathbf{x}_n$ of the stress variable, when the measurement times $t_1,...,t_k$ are predetermined.
The optimality criterion of minimization of the asymptotic variance of $\widehat{t}_\alpha$ thus reduces to a $c$-criterion $\mathbf{c}_{1}^T \mathbf{M}_{1}(\xi)^{-1} \mathbf{c}_{1}$ for extrapolation of the marginal response at normal use condition $\mathbf{x}_u$ in the first marginal model~(\ref{eq-marginal-1}), $\mathbf{c}_{1} = \mathbf{f}_1(\mathbf{x}_u)$, which is a well-known problem from the literature (see \cite{klefer1964optimum}).
It is remarkable that this criterion and, hence, the corresponding optimal design is the same whatever the value of $\alpha$ is, as long as there is a proper solution $0 < t_\alpha < \infty$ for the $\alpha$-quantile of the failure time.

\begin{proposition}
	\label{prop-extrapolation-x}
	If the design $\xi^*$ is $c$-optimal for extrapolation of the mean response at the normal use condition in the marginal model~(\ref{eq-marginal-1}) for the stress variable, then $\xi^*$ minimizes the asymptotic variance for the estimator $\widehat{t}_\alpha$ of the $\alpha$-quantile of the failure time for every $\alpha$ when $0 < t_\alpha < \infty$ (for predetermined measurement times $t_1, ..., t_k$).
\end{proposition}

Although the normal use condition is typically outside the experimental region, the above proposition also would hold for interpolation, i.\,e.\ $\mathbf{x}_u \in \mathcal{X}$.
The result of Proposition~\ref{prop-extrapolation-x} is next used to derive optimal designs for the situation in Examples~\ref{ex-intro} and \ref{ex-2}.

To quantify the quality of a standard design $\xi_0$ for estimating the quantile $t_\alpha$ of the mean failure time under normal use condition we make use of the efficiency
\begin{equation}
	\label{eq-avar-eff}
	\mathrm{eff}_{\mathrm{aVar}}(\xi_0) = \frac{\mathrm{aVar}(\widehat{t}_\alpha; \xi^*)}{\mathrm{aVar}(\widehat{t}_\alpha; \xi_0)} ,	
\end{equation} 
where $\mathrm{aVar}(\widehat{t}_\alpha; \xi) = \mathbf{c}^T \mathbf{M}_{\BS{\theta}}(\xi)^{-1} \mathbf{c}$ denotes the standardized asymptotic variance for estimating $t_\alpha$ by equation~(\ref{eq-info_theta-xi}) when design $\xi$ is used, and $\xi^*$ is the corresponding optimal design.
The efficiency gives the proportion of units to be used under the optimal design $\xi^*$ which provides (asymptotically) the same accuracy (in terms of the asymptotic variance) compared to the standard design $\xi_0$.
For example, if the efficiency is $0.5$ twice the number of units have to be used under $\xi_0$ than under the optimal design $\xi^*$ to get the same accuracy.
Note that both the asymptotic variance and the efficiency may also depend on the parameter vector $\BS\theta$, at least, through $t_\alpha$ and are, hence, local quantities (at $\BS\theta$) without explicitly stated in the notation.

In the case of estimating the median $t_{0.5}$ the standardized asymptotic variance factorizes as 
\begin{equation}
	\label{eq-avar-factorize}
	\mathrm{aVar}(\widehat{t}_{0.5}; \xi) = \frac{1}{n} c_0^2 \mathbf{f}_1(\mathbf{x}_u)^T \mathbf{M}_{1}(\xi)^{-1} \mathbf{f}_1(\mathbf{x}_u) \cdot \mathbf{f}_2(t_\alpha)^T \mathbf{M}_{2}^{-1} \mathbf{f}_2(t_\alpha)
\end{equation} 
by equations~(\ref{eq-avar-tau-median}) and (\ref{eq-avar-tau-alpha-decompose}) for the general product-type model~(\ref{full-model}).
Thus the efficiency defined in (\ref{eq-avar-eff}) reduces to the $c$-efficiency 
\begin{equation*}
	\label{eq-c-eff}
	\mathrm{eff}_c(\xi_0) = \frac{\mathbf{f}_1(\mathbf{x}_u)^T \mathbf{M}_{1}(\xi^*)^{-1} \mathbf{f}_1(\mathbf{x}_u)}{\mathbf{f}_1(\mathbf{x}_u)^T \mathbf{M}_{1}(\xi_0)^{-1} \mathbf{f}_1(\mathbf{x}_u)} 	
\end{equation*} 
for extrapolation at the normal use condition $\mathbf{x}_u$ in the first marginal model with uncorrelated homoscedastic errors and does not depend on $\BS\theta$. It has to be noted that the efficiency calculations for the numerical examples in section \ref{numeric-examples} are all related to estimating the median failure time for soft failure due to degradation under normal use conditions $\mathbf{x}_u$.
For estimating any other quantile $t_\alpha$ of the failure time distribution, the efficiency of a design $\xi$ can be written as 
\begin{equation}
	\label{eq-avar-x-lower-bound}
	\mathrm{eff}_{\mathrm{aVar}}(\xi) = \mathrm{eff}_c(\xi) + (1 - \mathrm{eff}_c(\xi)) \mathbf{c}_{\BS\varsigma}^T \widetilde{\mathbf{M}}_{\BS{\varsigma}}^{-1} \mathbf{c}_{\BS\varsigma} / \mathrm{aVar}(\widehat{t}_{\alpha}; \xi)
\end{equation}
by equations~(\ref{eq-avar-tau-alpha-sum}) and (\ref{eq-info_theta-xi}).
This efficiency depends on the variance parameters, but it is bounded from below by the $c$-efficiency $\mathrm{eff}_c(\xi)$ of $\xi$ for extrapolation at $\mathbf{x}_u$.
Hence, designs with a high efficiency for estimating the median failure time are also suitable for estimating any other reasonable quantile $t_\alpha$, $0 < t_\alpha < \infty$.
\section{Examples of optimal designs of stress variables}
\label{numeric-examples}
In this section we provide certrain examples of optimal designs of accelerated degradation testing.
We consider first a simple example based on \cite{doi:10.1002/asmb.2061}.

\begin{example}
 In Table~\ref{tab-ex-intro-nominal-values} we reproduce the nominal values of Example~7.2 by \cite{doi:10.1002/asmb.2061} on scar width growth after standardization for further use.

\begin{table}
	\begin{center}
		\caption{Nominal values for Example~\ref{ex-intro}}
		\label{tab-ex-intro-nominal-values}
		\vspace{2mm}
		\begin{tabular}{c|c|c|c|c||c|c|c|c||c||c}
			& $\beta_{1}$ & $\beta_2$ & $\beta_3$ & $\beta_4$ & $\sigma_{1}$ & $\sigma_{2}$ & $\rho$ & $\sigma_\varepsilon$ & ${x}_{u}$ & $y_0$
			\\
			$f_j(x,t)$ & $1$ & $x$ & $t$ & $x t$ & & & & & &
			\\
			\hline
			& $2.397$ & $1.629$ & $1.018$ & $0.0696$ & $0.114$ & $0.105$ & $-0.143$ & $0.048$ & $-0.056$ & $3.912$ 		
		\end{tabular} 
	\end{center}
\end{table}
The aggregate degradation path $\mu(t) = \delta_1 + \delta_2 t$ has intercept $\delta_1 = \beta_1 + \beta_2 x_u$ and slope $\delta_2 = \beta_3 + \beta_4 x_u$.
	Hence, the median failure time is given by $t_{0.5} = (y_0 - \beta_1 - \beta_2 x_u) / (\beta_3 + \beta_4 x_u)$.
	If we use the standardized nominal values of Table~\ref{tab-ex-intro-nominal-values}, the aggregate degradation path becomes $\mu(t) = 2.306 + 1.014 t$ under normal use condition, and the median failure time is $t_{0.5} = 1.583$.
	Note that, as typical for degradation experiments, the median failure time is larger than the maximal experimental time $t_{\max} = 1$. Subsequently, in view of equation \eqref{eq-h-tau-linear}, $h(t)$ is plotted in Figure~\ref{define-h-ex1} under the standardized nominal values of \citep{doi:10.1002/asmb.2061} given in Table~\ref{tab-ex-intro-nominal-values}.
	The defining function $h(t)$ is seen to be strictly increasing although the correlation is moderately negative ($\rho = - 0.143$). 
	Thus the distribution function $F_T(t) = \Phi(h(t))$ is well-defined, and it is represented in Figure~\ref{define-F-ex1}.
	In both plots the median failure time $t_{0.5} = 1.583$ is indicated by a dashed vertical line.
	Moreover, as $h(0) = - 14.03$ and $\lim_{t \to \infty} h(t) = 9.67$, the range of $h$ covers all reasonable quantiles.

\begin{figure}[!tbp]
  \centering
  \begin{minipage}[b]{0.4332\textwidth}
    \includegraphics[width=\textwidth]{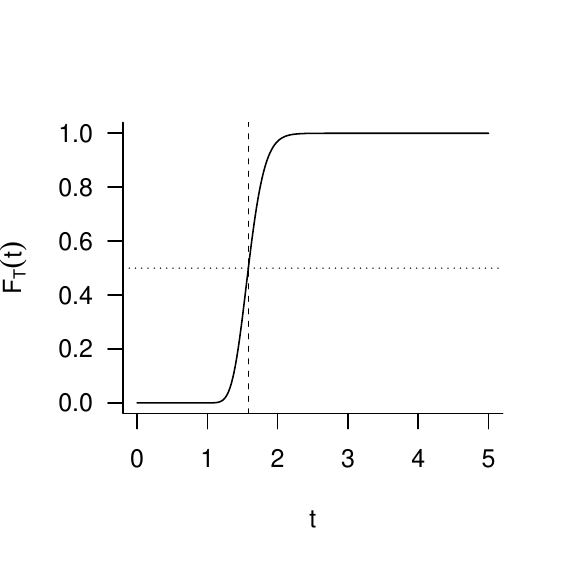}
    \caption{Failure time distribution $F_T$ for Example~\ref{ex-intro}}
\label{define-F-ex1}
  \end{minipage}
  \hfill
  \begin{minipage}[b]{0.43\textwidth}
    \includegraphics[width=\textwidth]{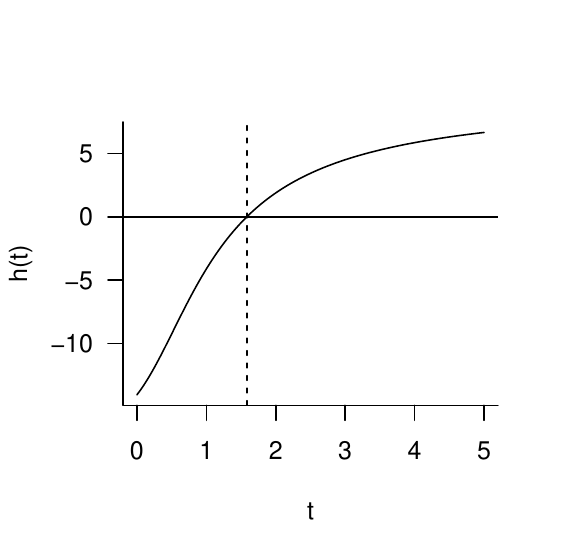}
   \caption{Defining function $h$ for Example~\ref{ex-intro}}
\label{define-h-ex1}
  \end{minipage}
\end{figure}
The marginal model for the stress variable $x$ is given by a simple linear regression, $\mathbf{f}_1(x) = (1, x)^T$.
	In this marginal model the $c$-criterion design $\xi^*$ for extrapolation of the mean response $\mu^{(1)}(x_u) = \beta_1^{(1)} + \beta_2^{(1)} x_u$ under normal use condition $x_u < 0$ is defined by $\Phi_c(\xi) = \mathbf{f}_1(x_u)^T \mathbf{M}_{1}(\xi)^{-1} \mathbf{f}_1(x_u)$. Accordingly, the $c$-optimal design $\xi^*$ assigns weight $w^* = |x_u| / (1 + 2 |x_u|)$ to the highest stress level $x_h = 1$ and weight $1 - w^* = (1 + |x_u|) / (1 + 2 |x_u|)$ to the lowest stress level $x_l = 0$ on the standardized scale $\mathcal{X} = [0, 1]$ $($see \citep{klefer1964optimum}$)$.
	Note that larger weight $1 - w^* > w^*$ is assigned to the lowest stress level $x_l = 0$ which is closer to $x_u < 0$ than $x_h = 1$ and that the weight $1 - w^*$ at $x_l$ decreases from $1$ to $1/2$, when the distance between the normal use condition and the experimental region gets larger, i.\,e.\ $x_u$ decreases.
	For the standardized value $x_u =- 0.056$ of the normal use condition from Table~\ref{tab-ex-intro-nominal-values} the optimal weights for extrapolation at $x_u =- 0.056$ are $w^* = 0.05$ at $x_h = 1$ and $1 - w^* = 0.95$ at $x_l = 0$, and the optimal design is
	\[
		\xi^* = \left(
		\begin{array}{cc}
			0 & 1
			\\
			0.95 & 0.05
		\end{array}
		\right) .
	\]
	Further examples for extrapolation at $x_{u} = - 0.4$, $ - 0.5$, and $ - 1$ give optimal weights $w^* = 0.22$, $0.25$, and $0.33$ at $x_h=1$, and $1 - w^* = 0.78$, $0.75$, and $0.67$ at $x_l=0$, respectively. By Proposition~\ref{prop-extrapolation-x} the design $\xi^*$ is also optimal for minimization of the asymptotic variance for estimating the $\alpha$-quantile $t_\alpha$ of the failure time for soft failure due to degradation under normal use condition $x_u$, when $0 < t_\alpha < \infty$ and the measurement times $t_1, ..., t_k$ are predetermined (see \citep{schwabe2014discussion} for estimation of the median, $\alpha = 0.5$).
	In particular, for the standardized value $x_u =- 0.056$ of the normal use condition from Table~\ref{tab-ex-intro-nominal-values} the optimal design for estimating any $\alpha$-quantile $t_\alpha$ assigns weight $0.95$ to $x_l = 0$ and weight $0.05$ to $x_h = 1$, as found numerically by \citep{doi:10.1002/asmb.2061} in the case of the median $t_{0.5}$.

 For $x_u<0$, the $c$-criterion at the present model
 attains its minimal value value $\Phi_c(\xi^*) = (1 + 2 |x_u|)^2$ for the optimal design $\xi^*$. 
	Common alternatives would be uniform designs $\bar\xi_m$ which assign equal weights $w = 1/m$ to $m$ experimental settings on an equidistant grid $\{x_1,...,x_m\} = \{0, 1/(m-1), ..., 1\}$ of the experimental region $\mathcal{X}=[0,1]$.
	For these designs the $c$-criterion for extrapolation at $x_u<0$ can be calculated as $\Phi_c(\bar\xi_m) = 1 + a_m (1 + 2 |x_u|)^2$, where $a_m = 3 (m - 1) / (m + 1)$.
	Their $c$-efficiency for extrapolation at $x_u$ and, hence, their efficiency for estimating the median failure time under normal use condition $x_u$ is equal to $\mathrm{eff}_c(\bar\xi_m) = \Phi_c(\xi^*) / \Phi_c(\bar\xi_m) = (1 - 1/(1 + a_m (1 + 2 |x_u|)^2))/a_m$ which increases from $1/(1 + a_m) = (m + 1) / (4 m - 2)$ for $x_u$ close to the lowest stress level $x_l=0$ to $1/a_m = (m + 1) / (3 m - 3)$ when $x_u$ tends to minus infinity.
	Moreover, for fixed $x_u$, the efficiency decreases when $m$ increases, i.\,e.\ when the grid becomes more dense. 
	For selected values of the normal use condition $x_u$ and numbers $m$ of grid points numerical values of the efficiency are reported in Table~\ref{tab-ex-intro-efficiency}
	\begin{table}
		\begin{center}
			\caption{Efficiency of uniform designs $\bar\xi_m$ for various normal use conditions $x_u$ in Example~\ref{ex-intro}}
			\label{tab-ex-intro-efficiency}
			\vspace{2mm}
			\begin{tabular}{c||c|cccc|c}
				$m$ & \multicolumn{6}{c}{$x_u$}
				\\
				& $\ \quad 0 \quad\ $ & $ - 0.056$ & $ - 0.400$ & $ - 0.500$ & $ - 1.000$ & $ \quad - \infty \quad$ 
				\\
				\hline
				$2$ & $0.50$ & $0.55$ & $0.76$ & $0.80$ & $0.90$ & $1.00$
				\\
				$3$ & $0.40$ & $0.43$ & $0.55$ & $0.57$ & $0.62$ & $0.67$
				\\
				$4$ & $0.36$ & $0.38$ & $0.47$ & $0.49$ & $0.52$ & $0.56$
				\\
				$5$ & $0.33$ & $0.36$ & $0.43$ & $0.44$ & $0.47$ & $0.50$
				\\
				\hline
				$\infty$ & $0.25$ & $0.26$ & $0.30$ & $0.31$ & $0.32$ & $0.33$
			\end{tabular} 
		\end{center}
	\end{table}
	Note that in Table~\ref{tab-ex-intro-efficiency} the row $m=\infty$ corresponds to a continuous uniform design as an approximation to large numbers $m$ of grid points, while the columns $x_u = 0$ and $x_u = - \infty$ give approximations for normal use conditions $x_u$ close to the lowest experimental stress level or far away, respectively.
	
	For the particular case $m = 2$, where the design $\bar{\xi}_2$ assigns equal weights $w = 1 - w = 1/2$ to both the highest and the lowest stress level $x_h = 1$ and $x_l = 0$, we have $a_2 = 1$ and, hence, $\Phi_c(\xi_0) = 1 + (1 + 2 |x_u|)^2$ for the $c$-criterion.
	The $c$-efficiency of $\bar{\xi}_2$ for extrapolation at $x_u$ and, thus, its efficiency for estimating the median failure time under normal use condition $x_u$ is equal to $\mathrm{eff}_c(\bar{\xi}_2) = 1 - 1 / (1 + (1 + 2 |x_u|)^2)$ which ranges from $1/2$ for $x_u$ close to the lowest stress level $x_l=0$ to $1$ when $x_u$ tends to minus infinity.
	
	For the nominal value $x_u = - 0.056$ of the normal use condition in Table~\ref{tab-ex-intro-nominal-values} the efficiency of the equidistant grid designs $\bar{\xi}_m$ is reported in the third column of Table~\ref{tab-ex-intro-efficiency}.
	In particular, for the uniform design $\bar{\xi}_2$ on the endpoints of the experimental region this efficiency is $0.55$ which means that $\mathrm{eff}_c(\bar{\xi}_2)^{-1} - 1 = 1/(1 + 2 |x_u|)^2 = 81\,\%$ more units would have to be used for design $\bar{\xi}_2$ to obtain the same quality for estimating the median failure time than for the optimal design $\xi^*$.
\end{example}


In the following we will consider thoroughly a more complex example, where two stress variables are involved under the virtual nominal values for the parameters, normal use conditions and threshold given in Table~\ref{tab-ex-2-nominal-values}.

 It has to be noted that in the case of the standardized nominal values of Table~\ref{tab-ex-intro-nominal-values} in Example~\ref{ex-intro} for high stress levels ($x_h = 1$) the mean degradation path $\mu_i$ exceeds the threshold $y_0$ for soft failure due to degradation with high probability ($\mathrm{P}(\mu_i(1, 0) \geq y_0) > 1 / 2$) already at the initial experimental time $t_{\min} = 0$.
Hence, care has to be taken that the model equation for the mean degradation paths is also valid beyond the threshold, i.\,e.\ in the case that soft failure has already occurred.
To avoid this complication we consider in Example~\ref{ex-2} nominal values which guarantee that soft failure occurs during the experiment only with negligible probability.

\begin{example}
	\label{ex-2}
In this example the degradation is influenced by two standardized accelerating stress variables $x_1$ and $x_2$ which act linearly on the response with a potential interaction effect associated with $x_1 x_2$.
The two stress variables $x_1$ and $x_2$ can be chosen independently from marginal design regions  $\mathcal{X}_1 = \mathcal{X}_2 = [0, 1]$, respectively. 
Also the time is assumed to act linearly on the degradation and all interactions between stress variables and time are present as in Example~1.
\begin{table}
	\begin{center}
		\caption{Nominal values for Example~\ref{ex-2}}
		\label{tab-ex-2-nominal-values}
		\vspace{2mm}
		\begin{tabular}{c|c|c|c|c|c|c|c|c||c|c||c|c||c}
			& $\beta_{1}$ & $\beta_2$ & $\beta_3$ & $\beta_4$ 
			& $\beta_{5}$ & $\beta_6$ & $\beta_7$ & $\beta_8$ 
			& $\sigma_\gamma$ & $\sigma_\varepsilon$ 
			& ${x}_{u 1}$ & ${x}_{u 2}$ 
			& $y_0$
			\\
			$f_j(\mathbf{x}, t)$ 
			& $1$ & $x_1$ & $x_2$ & $x_1 x_2$ 
			& $t$ & $x_1 t$ & $x_2 t$ & $x_1 x_2 t$ & & & & &
			\\
			\hline
			& $4.0$ & $1.5$ & $0.75$ & $1.8$ 
			& $0.5$ & $0.25$ & $0.25$ & $4.03$ 
			& $0.7$ & $0.85$ 
			& $ - 0.5$ & $ - 0.4$ 
			& $14.39$
		\end{tabular} 
	\end{center}
\end{table}
If, for testing unit $i$, the stress variables are set to $x_{i 1}$ and $x_{i 2}$ the response $y_{i j}$ at time $t_j$ is given by
\begin{equation} 
\label{example2_1}
 y_{i j} = \beta_{i, 1} + \beta_2 x_{i 1} + \beta_3 x_{i 2} + \beta_4 x_{i 1} x_{i 2} + \beta_{i, 5} t_j + \beta_6 x_{i 1} t_j + \beta_7 x_{i 2} t_j + \beta_8 x_{i 1} x_{i 2} t_j + \varepsilon_{i j} ,
\end{equation}
where the intercept $\beta_{i, 1}$ is the mean degradation of unit $i$ at time $t = 0$ under the stress levels $x_1 = 0$ and $x_2 = 0$, $\beta_{2}$ is the common (not unit specific) mean increase in degradation depending on the stress variable $x_1$ when $x_2 = 0$, $\beta_{3}$ is the common mean increase in degradation depending on the stress variable $x_2$ when $x_1 = 0$, and $\beta_{4}$ is the interaction effect between the two stress variables.
Accordingly $\beta_{i, 5}$ is the mean increase in degradation of unit $i$ over time $t$ when the stress levels are set to $x_1 = 0$ and $x_2 = 0$, $\beta_{6}$ is the interaction effect between  time and the stress variable $x_1$ when $x_2 = 0$, $\beta_{7}$ is the interaction effect between  time and the stress variable $x_2$ when $x_1 = 0$, and $\beta_{8}$ is the second-order interaction effect between time and the two stress variables.
Also here only the parameters $\beta_{i,1}$ and $\beta_{i,5}$ associated with the constant term in the stress variables may vary across units.
On the aggregate level these two unit parameters are again assumed to be normally distributed with means $\mathrm{E}(\beta_{i, 1})=\beta_{1}$ and $\mathrm{E}(\beta_{i, 5})=\beta_{5}$ and variance covariance matrix $\BS{\Sigma} = \left(\begin{array}{cc} \sigma^{2}_{1} & \rho \sigma_{1} \sigma_{2}  \\ \rho\sigma_{1} \sigma_{2} & \sigma_{2}^{2} \end{array}\right)$.
After rearranging terms and relabeling the parameters the model can be rewritten as
\begin{equation} 
	\label{example2_1-formal}
	Y_{i j} = (\mathbf{f}_{1 1}(x_{i 1}) \otimes \mathbf{f}_{1 2}(x_{i 2}) \otimes \mathbf{f}_2(t_j))^{T} \BS{\beta} + \mathbf{f}_2(t_j)^{T} \BS\gamma_i + \varepsilon_{i j},
\end{equation}
where $\mathbf{f}_{11}(x_1) = (1, x_1)^T$, $\mathbf{f}_{12}(x_2) = (1, x_2)^T$ and $\mathbf{f}_2(t) = (1, t)^T$ are the marginal regression functions for the stress variables $x_1$, $x_2$ and the time variable $t$, respectively, $\BS\beta = (\beta_{1 1 1}, \beta_{1 1 2}, \beta_{1 2 1}, \beta_{1 2 2}, \beta_{2 1 1}, \beta_{2 1 2},\\ \beta_{2 2 1}, \beta_{2 2 2})^T = (\beta_{1}, \beta_{5}, \beta_{3}, \beta_{7}, \beta_{2}, \beta_{6}, \beta_{4}, \beta_{8})^T$ is the rearranged vector of aggregate parameters, and $\BS\gamma_i = (\gamma_{i 1}, \gamma_{i 2})^T = (\beta_{i,1} - \beta_{1}, \beta_{i,5} - \beta_{5})^T$ is the vector of parameters for the deviations of unit $i$ from the aggregate values.
These deviations constitute again random effects with zero mean and variance covariance matrix $\BS\Sigma$.
With $\mathbf{x} = (x_1, x_2)$, $\mathcal{X} = \mathcal{X}_1 \times \mathcal{X}_2 = [0,1]^2$,  $\mathbf{f}_1(\mathbf{x}) = (\mathbf{f}_{1 1}(x_1)^T, \mathbf{f}_{1 2}(x_2)^T)^T$, $p_1 = 4$, $p_2 = 2$, and $p = 8$ model~\ref{example2_1-formal} fits into the framework of the general product-type model~\ref{modelindividuallevel-product}.

	The aggregate degradation path $\mu(t)$ is a straight line with intercept $\delta_1 = \beta_1 + \beta_2 x_{u 1} + \beta_3 x_{u 2} + \beta_4 x_{u 1} x_{u 2} = 3.31$ and slope $\delta_2 = \beta_5 + \beta_6 x_{u 1} + \beta_7 x_{u 2} + \beta_8 x_{u 1} x_{u 2} = 1.08$, i.\,e.\ $\mu(t) = 3.31 + 1.08 t$ under normal use condition.
	With a threshold of $y_0 = 14.39$ for soft failure the median failure time results in $t_{0.5} = (y_0 - \delta_1) / \delta_2 = 10.25$ which is substantially larger than the maximal experimental time $t_{\max} = 1$.
	For the characterization of other quantiles the function $h(t)$ is plotted in Figure~\ref{define-h-ex2} together with the corresponding distribution function $F_T(t)=\Phi(h(t))$ in Figure~\ref{define-F-ex2}.
	The median failure time $t_{0.5} = 10.25$ is indicated in both plots by a dashed vertical line.
	As $\rho = 0$ the function $h(t)$ is strictly increasing and ranges from $h(0) = - 15.83$ to $h_{\max} = \lim_{t\to\infty} h(t) = 1.54$. 
	Thus, quantiles $t_\alpha$ are non-degenerate as long as $\alpha \leq \alpha_{\max}$, where $\alpha_{\max} = \Phi(h_{\max}) = 0.939$, and $(1 - \alpha_{\max}) \cdot 100 = 6.1$ percent of the mean degradation paths do not lead to a soft failure.
	Both $\alpha_{\max}$ and $h_{\max}$ are indicated in the respective plots by a dashed horizontal line.
\begin{figure}[!tbp]
  \centering
  \begin{minipage}[b]{0.44\textwidth}
    \includegraphics[width=\textwidth]{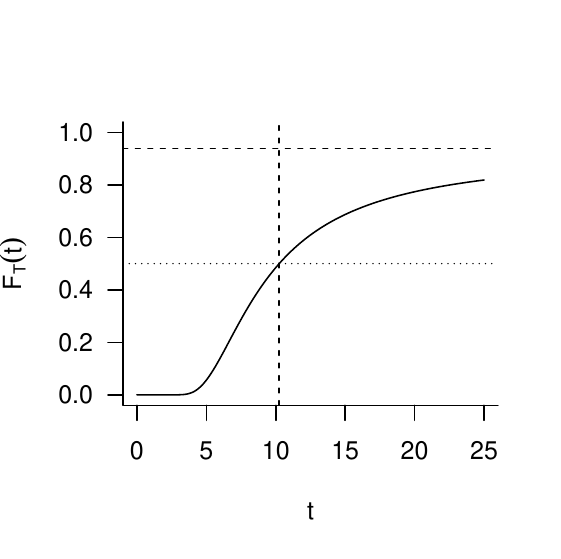}
    \caption{Failure time distribution $F_T$ for Example~\ref{ex-2}}
\label{define-F-ex2}
  \end{minipage}
  \hfill
  \begin{minipage}[b]{0.425\textwidth}
    \includegraphics[width=\textwidth]{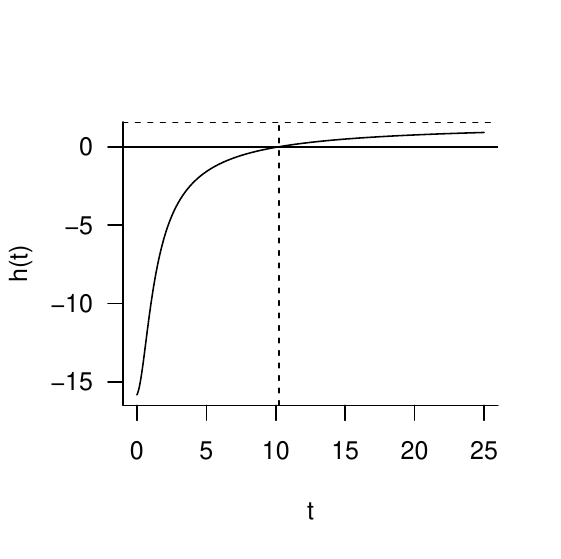}
   \caption{Defining function $h$ for Example~\ref{ex-2}}
\label{define-h-ex2}
  \end{minipage}
\end{figure}
	Note that for the nominal values of Table~\ref{tab-ex-2-nominal-values} the mean degradation $\mu(\mathbf{x},t)$ under experimental conditions attains its maximum $13.08$ for the maximal stress levels ($x_1 = x_2 = 1$) and maximal experimental time ($t = 1$), and, hence, the mean degradation paths do not exceed the threshold for all experimental settings.

The marginal model for the present combined stress variable $\mathbf{x} = (x_1, x_2)$ is given itself by a product-type structure, $\mathbf{f}_1(\mathbf{x}) = \mathbf{f}_{1 1}(x_1) \otimes \mathbf{f}_{1 2}(x_2)$, where both components $x_1$ and $x_2$ are specified by $\mathbf{f}_{1 v}(x_v) = (1,x_v)^T$ as simple linear regressions in their corresponding submarginal models $Y_{i}^{(1 v)} = \beta_1^{(1 v)} + \beta_2^{(1 v)} x_v + \varepsilon_i^{(1 v)}$, $v = 1, 2$, with standardized homoscedastic and uncorrelated error terms.
	Moreover, the experimental region $\mathcal{X} = [0, 1]^2$ for the combined stress variable $\mathbf{x}$ is the Cartesian product of the marginal experimental regions $\mathcal{X}_1 = \mathcal{X}_2 = [0, 1]$ for the components $x_1$ and $x_2$, respectively.
	The vector $\mathbf{c}$ for extrapolation of the mean response $\mathbf{c}^T\BS\beta^{(1)} = \mu^{(1)}(\mathbf{x}_u) = \beta_1^{(1)} + \beta_2^{(1)} x_{u 1} + \beta_3^{(1)} x_{u 2} + \beta_4^{(1)} x_{u 1}x_{u 2}$ under normal use condition $\mathbf{x}_u=(x_{u 1},x_{u 2})$, $x_{u 1},x_{u 2}<0$, is given by $\mathbf{f}_1(\mathbf{x}_u) = \mathbf{f}_{1 1}(x_{u 1}) \otimes \mathbf{f}_{1 2}(x_{u 2})$ and, hence, also factorizes as $\mathbf{c} = \mathbf{c}_1 \otimes \mathbf{c}_2$, where $\mathbf{c}_v = \mathbf{f}_{1 v}(x_{u v})$.
	In this setting the $c$-optimal design $\xi^*$ for extrapolation at $\mathbf{x}_u$ can be obtained as the product $\xi^* = \xi_1^* \otimes \xi_2^*$ of the $c$-optimal designs $\xi_v^*$  for extrapolation at $x_{u v}$ in the submarginal models (see Theorem~4.4 in \citep{schwabe1996optimum}). 
	
	The submarginal $c$-optimal designs $\xi_v^*$ can be derived as in Example~\ref{ex-intro}.
	They assign weight $w_v^* = |x_{u v}| / (1 + 2|x_{u v}|)$ to $x_v = 1$ and weight $1 - w_v^* = (1 + |x_{u v}|) / (1 + 2|x_{u v})$ to $x_v = 0$.
	Hence, the $c$-optimal design $\xi^* = \xi_1^* \otimes \xi_2^*$ for extrapolation at $\mathbf{x}_u$ is given by
	\[
	\xi^* = \left(
		\begin{array}{cccc}
			(0, 0) & (0, 1) & (1, 0) & (1, 1)
			\\
			(1 - w_1^*) (1 - w_2^*) & (1 - w_1^*) w_2^* & w_1^* (1 - w_2^*) & w_1^* w_2^* 
		\end{array}
	\right) .
	\]
	
	Then, by Proposition~\ref{prop-extrapolation-x}, the design $\xi^*$ is also optimal for minimization of the asymptotic variance for estimating the $\alpha$-quantile $t_\alpha$ of the failure time for soft failure due to degradation, when $0 < t_\alpha < \infty$ and the measurement times $t_1, ..., t_k$ are predetermined.
	For example, when the normal use conditions are ${x}_{u 1} = - 0.5$ for the first component and ${x}_{u 2} = - 0.4$ for the second component as specified in Table~\ref{tab-ex-2-nominal-values}, then by the results in Example~\ref{ex-intro} the optimal marginal weights are $w_1^* = 0.25$ and $w_2^* = 0.22$, and the optimal design $\xi^* = \xi_1^* \otimes \xi_2^*$ is given by
	\[
	\xi^* = \left(
		\begin{array}{cccc}
			(0, 0) & (0, 1) & (1, 0) & (1, 1)
			\\
			0.58 & 0.17 & 0.19 & 0.06 
		\end{array}
		\right) .
	\]

The corresponding $c$-criterion $\Phi_c(\xi) = \mathbf{f}_1(\mathbf{x}_u)^T \mathbf{M}_{1}(\xi)^{-1} \mathbf{f}_1(\mathbf{x}_u)$ for extrapolation at $\mathbf{x}_u = (x_{u 1}, x_{u 2})$ factorizes into its counterparts in the submarginal models, $\Phi_c(\xi) = \mathbf{f}_{1 1}({x}_{u 1})^T \mathbf{M}_{1 1}(\xi_1)^{-1}  \mathbf{f}_{1 1}({x}_{u 1}) \cdot \mathbf{f}_{1 2}({x}_{u 2})^T \mathbf{M}_{1 2}(\xi_2)^{-1}  \mathbf{f}_{1 2}({x}_{u 2})$.
	Because also the $c$-optimal design $\xi^* = \xi_1^* \otimes \xi_2^*$ has product-type structure, the $c$-efficiency for extrapolation at $\mathbf{x}_u$ and, hence, the efficiency for estimating the median failure time factorizes, $\mathrm{eff}_c(\xi_1\otimes\xi_2) = \mathrm{eff}_{c 1}(\xi_1) \cdot \mathrm{eff}_{c 2}(\xi_2)$, where $\mathrm{eff}_{c v}(\xi_v)$ is the corresponding efficiency in the $v$th submarginal model, $v = 1, 2$.
	
	The design $\bar\xi$ which assigns equal weights $1/4$ to the four vertices $(0, 0)$, $(0, 1)$, $(1, 0)$, and $(1, 1)$ of the experimental region serves as a natural standard design.
	This design can be seen to be the product $\bar\xi = \bar\xi_2 \otimes \bar\xi_2$ of submarginal designs $\bar{\xi}_2$ which assign equal weights $1 / 2$ to the lowest and highest stress level $x_{v l}$and $x_{v h}$  in the submarginal models, $v = 1, 2$.
	Hence, from Example~\ref{ex-intro} we get the efficiency of $\bar\xi$ as $\mathrm{eff}_c(\bar\xi) = ((1 + 2 |x_{u 1}|)(1 + 2 |x_{u 2}|))^2/((1 + (1 + 2 |x_{u 1}|)^2)(1 + (1 + 2 |x_{u 2}|)^2))$ which ranges from $1/4$ for $\mathbf{x}_u$ close to the combination $(x_{1 l}, x_{2 l})$ of lowest stress levels $x_{1 l} = 0$ and $x_{2 l} = 0$ to $1$ when both normal use conditions $x_{u 1}$  and $x_{u 2}$ tend to minus infinity.

	For example, when the normal use conditions are ${x}_{u 1} = - 0.5$ for the first component and ${x}_{u 2} = - 0.4$ for the second component as specified in Table~\ref{tab-ex-2-nominal-values}, then according to Table~\ref{tab-ex-intro-efficiency} the efficiency $\mathrm{eff}_{c v}(\bar\xi_2)$ of $\bar{\xi}_2$ is $0.80$ and $0.76$ in the respective submarginal models, $v = 1, 2$.
	By the above considerations the efficiency of $\bar\xi$ is $\mathrm{eff}_c(\bar{\xi}) = 0.80 \cdot 0.76 = 0.61$.
	This means that $\mathrm{eff}_c(\bar{\xi})^{-1} - 1 = 0.39 / 0.61 = 64\,\%$ more units have to be used for design $\bar{\xi}$ to obtain the same quality for estimating the median failure time than for the optimal design $\xi^*$.
	Hence, the optimal design $\xi^*$ performs much better than the standard design $\bar\xi$ in this situation.	
\end{example}

In Example \ref{ex-3-additive} we will use Elfving's theorem to characterize optimal designs for the situation with two non-interacting stress variables. Further details in regards to Elfving's theorem are deffered the Appendix.

\begin{example}
	\label{ex-3-additive}
	In the case of two non-interacting stress variables $x_1$ and $x_2$ we consider the model equation	
	\begin{equation} 
		\label{example3_1}
		y_{i j} = \beta_{i, 1} + \beta_2 x_{i 1} + \beta_3 x_{i 2}  + \beta_{i, 4} t_j + 	\beta_5 x_{i 1} t_j + \beta_6 x_{i 2} t_j+ \varepsilon_{i j} 
	\end{equation}
	for the combined stress variable $\mathbf{x}=(x_1,x_2)$ and the time variable $t$.
	This model contains all terms of the full interaction model~\ref{example2_1}) with the exception of the terms $x_1 x_2$ and $x_1 x_2 t$ related to potential interactions between the stress variables.
	The interpretation of all other terms in (\ref{example3_1}) is the same as in Example~\ref{ex-2}.
	Model~(\ref{example3_1}) is constructed from the marginal model  
	\[
		Y_{i}^{(1)} = \beta_1^{(1)} + \beta_2^{(1)} x_{i 1} + \beta_3^{(1)} x_{i 2} + \varepsilon_i^{(1)} 
	\]
	which is additive in the effects of the stress variables $x_1$ and $x_2$,
	i.\,e.\ $\mathbf{f}_1(\mathbf{x}) = (1, x_1, x_2)^T$. 
	
	For this marginal model of multiple regression, the Elfving set is an oblique prism with quadratic base with vertices $(1, 0, 0)^T$, $(1, 0, 1)^T$, $(1, 1, 0)^T$, $(1, 1, 1)^T$ and quadratic top with vertices $( - 1, 0, 0)^T$, $( - 1, 0, - 1)^T$, $( - 1, - 1, 0)^T$, $( - 1, - 1, - 1)^T$.
	To find the $c$-optimal extrapolation design at the normal use condition $\mathbf{x}_u=(x_{u 1},x_{u 2})$ by Elfving's theorem we have to determine the intersection point of the ray $\lambda(1, x_{u 1}, x_{u 2})^T$ with the surface of the Elfving set.
	For $x_{u 1} < x_{u 2} < 0$ the ray intersects the surface at the quadrangular face of the prism spanned by $(1, 0, 0)^T$, $(1, 0, 1)^T$, $( - 1, - 1, 0)^T$, and $( - 1, - 1, - 1)^T$ when $\lambda_c = 1 / (1 + 2 |x_{u 1}|)$.
	The representation of the intersection point $\lambda_c(1,x_{u 1},x_{u 2})^T$ by the vertices of the quadrangle is not unique.
	There are two $c$-optimal designs
	\[
		\xi_0^* = \left(
		\begin{array}{ccc}
			(0, 0) & (0, 1) & (1, 1)
			\\
			(1 + |x_{u 2}|) \lambda_c & (|x_{u 1}| - |x_{u 2}|) \lambda_c & |x_{u 1}| \lambda_c
		\end{array}
		\right)	
	\]
	and
	\[
		\xi_1^* = \left(
		\begin{array}{ccc}
			(0, 0) & (1, 0) & (1, 1)
			\\
			(1 + |x_{u 1}|) \lambda_c & (|x_{u 1}| - |x_{u 2}|) \lambda_c & |x_{u 2}| \lambda_c
		\end{array}
		\right)	
	\]
	which are supported on three vertices.
	As a consequence, also for all coefficients $a$, $0 < a < 1$, the convex combination
\begin{align*}
	& \xi_{a}^* = (1 - a) \xi_0^* + a \xi_1^* =
		\\
		& \left(	
		\begin{array}{cccc}
			(0, 0) & (0, 1)  & (1, 0) & (1, 1)
			\\
			(1 + a |x_{u 1}| + (1 - a) |x_{u 2}|) \lambda_c & (1 - a) (|x_{u 1}| - |x_{u 2}|) \lambda_c & a (|x_{u 1}| - |x_{u 2}|) \lambda_c & ((1 - a) |x_{u 1}| + a |x_{u 2}|) \lambda_c
		\end{array}
		\right)	
\end{align*}	
supported on all four vertices is $c$-optimal for extrapolation at $\mathbf{x}_u$, $0 < t_\alpha < \infty$.
	Then, by Proposition~\ref{prop-extrapolation-x}, the designs $\xi_{a}^*$ are also optimal for minimization of the asymptotic variance for estimating the $\alpha$-quantile $t_\alpha$ of the failure time for soft failure due to degradation, when $0 < t_\alpha < \infty$ and the measurement times $t_1, ..., t_k$ are predetermined, $0 \leq \alpha \leq 1$.
	For example, when the normal use conditions are $x_{u 1} = - 0.5$ for the first component and $x_{u 2} = - 0.4$ for the second component as in Example~\ref{ex-2}, then the optimal design $\xi_{a}^*$ is given by
	\[
		\xi_{a}^* = \left(
		\begin{array}{cccc}
			(0, 0) & (0, 1) & (1, 0) & (1, 1)
			\\
			0.70 + 0.05 a & 0.05 - 0.05 a & 0.05 a & 0.25 - 0.05 a
		\end{array}
		\right) 
	\]
	with the special cases
	\[
		\xi_0^* = \left(
		\begin{array}{ccc}
			(0, 0) & (0, 1) & (1, 1)
			\\
			0.70 & 0.05 & 0.25
		\end{array}
		\right)	
		\qquad \mathrm{and} \qquad
		\xi_1^* = \left(
		\begin{array}{ccc}
			(0, 0) & (1, 0) & (1, 1)
			\\
			0.75 & 0.05  & 0.20
		\end{array} 
		\right)	
	\]
	supported on three vertices.
	
	Note that there are also other designs which are $c$-optimal for extrapolation at $\mathbf{x}_u$, but which are not solely supported on the vertices. 
	For example, for $x_{u 1} < x_{u 2} < 0$ the two-point design which assigns weight $w = |x_{u 1}| / (1 + 2 |x_{u 1}|)$ to $(0, 0)$ and weight $1 - w = (1 + |x_{u 1}|) / (1 + 2 |x_{u 1}|)$ to $(1, x_{u 2} / x_{u 1})$ is $c$-optimal by Elfving's theorem.
	However, these designs can be used for estimating $t_\alpha$ by means of maximum-likelihood only when the resulting information matrix is non-singular, i.\,e.\ when the design has, at least, three distinct support points.
	
	For $x_{u 2} < x_{u 1} < 0$ optimal designs can be obtained from the above case by interchanging the roles of the two components $x_1$ and $x_2$.
	
	In the case $x_{u 1} = x_{u 2} < 0$ there is only one $c$-optimal design for extrapolation.
	This design is supported on two vertices and assigns weight $w = |x_{u 1}|/(1 + 2 |x_{u 1}|)$ to $(0,0)$ and weight $1 - w = (1 + |x_{u 1}|)/(1 + 2 |x_{u 1}|)$ to $(1,1)$.
	As the resulting information matrix is singular, this design cannot be used for estimating the $\alpha$-quantile $t_\alpha$ of the failure time for soft failure due to degradation. 
	Hence, no suitable optimal design exists in this case, but the $c$-optimal design may serve as a benchmark for judging the quality of a competing design in terms of efficiency.

the value of the $c$-criterion for the locally $c$-optimal design $\xi^*$ for extrapolation at $\mathbf{x}_u = (x_{u 1}, x_{u 2})$, $x_{u 1} < x_{u 2} < 0$, is given by $\Phi_c(\xi^*) = 1 / \lambda_c^2 = (1 + 2 |x_{u 1}|)^2$ as seen before.
	The uniform design $\bar\xi$ which assigns equal weights $1/4$ to the four vertices $(0, 0)$, $(0, 1)$, $(1, 0)$, and $(1, 1)$ of the experimental region has a value of $\Phi_c(\bar\xi) = 1 + (1 + 2 |x_{u 1}|)^2 + (1 + 2 |x_{u 2}|)^2$.
	Hence, the uniform design $\bar\xi$ has efficiency $\mathrm{eff}_c(\bar\xi) = ((1 + 2 |x_{u 1}|)^2 / (1 + (1 + 2 |x_{u 1}|)^2 + (1 + 2 |x_{u 2}|)^2)$ which ranges from $1 / 3$ for $\mathbf{x}_u$ close to the combination $(x_{1 l}, x_{2 l})$ of lowest stress levels $x_{1 l} = 0$ and $x_{2 l} = 0$ to $1$ when the lower normal use condition $x_{u 1}$ tends to minus infinity while $x_{u 2}$ remains fixed.
	Moreover, the efficiency approaches $1 / 2$ when $x_{u 2} \approx x_{u 1}$ and both normal use conditions tend to minus infinity simultaneously.
	
	For example, when the normal use conditions are ${x}_{u 1} = - 0.5$ for the first component and ${x}_{u 2} = - 0.4$ for the second component as specified in Table~\ref{tab-ex-2-nominal-values}, then the values of the $c$-criterion are $\Phi_c(\xi^*) = 4.00$ for the optimal design $\xi^*$ and $\Phi_c(\bar\xi) = 8.24$ for the uniform design $\bar\xi$, respectively.
	Hence, the efficiency of the uniform design $\bar{\xi}$ is $\mathrm{eff}_{c}(\bar\xi) = 4.00 / 8.24 = 0.49$.
	This means that more than twice as many units have to be used for design $\bar{\xi}$ to obtain the same quality for estimating the median failure time than for the optimal design $\xi^*$.
	This highlights that the optimal design $\xi^*$ performs substantially better than the standard design $\bar\xi$ in the current model of two non-interacting stress variables.
\end{example}

\section{Discussion and conclusion}
\label{sec-discussion}

During the design stage of highly reliable systems it is extremely important to assess the reliability related properties of the product. 
One method to handle this issue is to conduct accelerated degradation testing. 
Accelerated degradation tests have the advantage to provide an estimation of lifetime and reliability of the system under study in a relatively short period of time. 
To account for variability between units in accelerated degradation tests, it is assumed that the degradation function can be described by a mixed-effects linear model.
This also leads to a non-degenerate distribution of the failure time, due to soft failure by exceedance of the expected (conditionally per unit) degradation path over a threshold, under normal use conditions.
Therefore it is desirable to estimate certain quantiles of this failure time distribution as a characteristic of the reliability of the product.
In this context we discussed the existence of non-degenerate solutions for the quantiles.
The purpose of optimal experimental design is then to find the best settings for the stress variable and/or the time variable to obtain most accurate estimates for these quantities.

In the present model for accelerated degradation testing, it is further assumed that stress remains constant within each testing unit during the whole period of experimental measurements but may vary between units.
Hence, in the corresponding experiment a cross-sectional design between units has to be chosen for the stress variable while for repeated measurements the time variable varies according to a longitudinal design within units.

In the present paper we assumed a model with complete interactions between the time and the stress variables and random effects only associated with time but not with stress.
Then the cross-sectional design for the stress variables and the longitudinal design for the time variable can be optimized independently, and the resulting common optimal design can be generated as the cross-product of the optimal marginal designs for stress and time, respectively.
In particular, the same time plan for measurements can be used for all units in the test.
Moreover, the marginal optimal design for the stress variables can be chosen independently of any model parameters.
Optimal time plans may depend on the aggregate location parameters via the median failure time, but do not depend on which quantile of the failure distribution is to be estimated.
These results were extended to a model of destructive testing in which also the time variable has to be chosen cross-sectionally.
There the optimal choice of measurement times may also be affected by the variance covariance parameters of the random effects.
In both cases (longitudinal and cross-sectional time settings) the efficiency of the designs considered factorizes which facilitates to assess their performance when the nominal values for these parameters are misspecified at the design stage.

Finding optimal designs may become more complicated when the above assumptions are not met.
In particular, the designs for stress and time variables may no longer be optimized independently if there are only additive effects in the model (lacking interaction terms $x t$, cf.\ Example~\ref{ex-3-additive} for a similar situation in the marginal stress model) or when also the stress variables are accompanied by random effects. 
The impact of these deviations from the model assumptions on optimal designs are object of further research as well as the construction of designs which are robust against misspecification of the nominal parameters, such as maximin efficient or weighted (``Bayesian'') optimal designs.
Of further interest would be to consider optimality criteria accounting for simultaneous estimation of various characteristics of the failure time distribution. 
\renewcommand{\thesection}{\Alph{section}}
\setcounter{section}{0}

\section{Appendix: Elfving's theorem}
\label{Elv.theor.}
The $c$-optimal extrapolation designs can be obtained in both Examples~\ref{ex-intro} and \ref{ex-2} by Elfving's theorem (\citep{elfving1952}) which provides a geometrical construction of a $c$-optimal design (see \citep{schwabe1996optimum}, Theorem~2.13).
To give a rough idea of this construction one has to consider the Elfving set which is the convex hull 
\[
\mathcal{E} = \mathrm{conv}(\{\mathbf{f}(\mathbf{x});\, \mathbf{x} \in \mathcal X\} \cup \{ - \mathbf{f}(\mathbf{x});\, \mathbf{x} \in \mathcal X\})
\] 
of the union of the so-called induced design region $\{\mathbf{f}(\mathbf{x});\, \mathbf{x}\in \mathcal X\}$ and its image $\{-\mathbf{f}(\mathbf{x});\, \mathbf{x}\in \mathcal X\}$ under reflection at the origin $\mathbf{0}$ in $\mathbb{R}^p$.
Here $\mathbf{x}$ and $\mathbf{f}$ denote variables and regression functions associated with a generic model $Y_i = \mathbf{f}(\mathbf{x}_i)^T \BS\beta + \varepsilon_i$.
In Example~\ref{ex-intro} we have $\mathbf{x} = x$, $\mathbf{f}(\mathbf{x}) = (1, x)^T$ and $\mathcal{X} = [0, 1]$, and the Elfving set $\mathcal{E}$ is given as a parallelogram in $\mathbb{R}^2$ with one edge from $(1, 0)^T$ to $(1, 1)^T$ representing the induced design region and the opposite edge from $( - 1, 0)^T$ to $( - 1, - 1)^T$ representing its image under reflection.

The $c$-optimal design for estimating $\mathbf{c}^T\BS\beta$ can then be constructed as follows:
Determine the intersection point of the ray $\lambda\mathbf{c}$, $\lambda>0$, with the boundary of the Elfving set, $\lambda_c \mathbf{c}$ say.
This point can be represented as a convex combination of (extremal) points $\pm \mathbf{f}(\mathbf{x}_i)$ of the induced design region and its reflection,
\begin{equation*}
	\lambda_c \mathbf c = \sum_{i = 1}^m w_i z_i \mathbf{f}(\mathbf{x}_i) ,
	\label{eq-th-elfving}
\end{equation*}
where $z_i = 1$, when the (extremal) point $\mathbf{f}(\mathbf{x}_i)$ is from the induced design region, and $z_i = - 1$, when the point $-\mathbf{f}(\mathbf{x}_i)$ is from the reflection, and the weights $w_i$ of the convex combination satisfy $w_i > 0$ and $\sum_{i = 1}^m w_i = 1$.
Then Elfving's theorem states that the design $\xi^*$ which assigns weights $w_i$ to the settings $\mathbf{x}_i$ is $c$-optimal (for $\mathbf{c}$).
Moreover, this construction also provides the value of the $c$-criterion, $\mathbf{c}^T \mathbf{M}
(\xi^*)^{-1} \mathbf{c} = 1 / \lambda_c^2$.
In Example~\ref{ex-intro} the ray $\lambda(1, x_u)^T$ intersects the boundary of the Elfving set $\mathcal{E}$ at the connecting line from $(1, 0)^T=\mathbf{f}(0)$ to $( - 1, - 1)^T = - \mathbf{f}(1)$ at  $\lambda_c (1, x_u)^T = w_1( - 1, - 1)^T + w_2(1, 0)^T$ with $\lambda_c = 1/(1 - 2 x_u)$, $w_1= - x_u / (1 - 2 x_u) > 0$ and $w_2 = 1 - w_1 = (1 - x_u) / (1 - 2 x_u) > 0$.
 Hence, the optimality of the given design follows.

\section*{Acknowledgement}
The work of the first author has been supported by the German Academic Exchange Service (DAAD) under grant no.~2017-18/ID-57299294.

\section*{References}
 \bibliographystyle{ieeetr}
  \bibliography{ReferencesLMEM}

\end{document}